\documentclass[12pt]{article}
\usepackage{float}
\usepackage[english]{babel}
\usepackage[margin=1in]{geometry}
\usepackage{amssymb}
\usepackage{amsthm}
\usepackage{amsmath, amsfonts}
  \usepackage{paralist}
  \usepackage{graphics} %% add this and next lines if pictures should be in esp format
  \usepackage{epsfig} %For pictures: screened artwork should be set up with an 85 or 100 line screen
\usepackage{graphicx}  
\usepackage{epstopdf}%This is to transfer .eps figure to .pdf figure; please compile your paper using PDFLeTex or PDFTeXify.
\usepackage[font={small,it}]{caption}
\usepackage{subcaption}
\graphicspath{{figures/}} % Location of the graphics files

 \usepackage[colorlinks=true]{hyperref}
\hypersetup{urlcolor=blue, citecolor=red}
\DeclareMathOperator{\spn}{span}

\DeclareMathOperator{\Id}{Id}
\DeclareMathOperator{\R}{R}
\DeclareMathOperator{\N}{N}

\newtheorem{theorem}{Theorem}[section]

\newtheorem{lemma}[theorem]{Lemma}
\newtheorem{proposition}[theorem]{Proposition}
\newtheorem{conjecture}[theorem]{Conjecture}
\theoremstyle{definition}
\newtheorem{definition}[theorem]{Definition}
\newtheorem{remark}[theorem]{Remark}

\numberwithin{equation}{section}
\begin{document}

{\centering

{\bfseries\Large Numerical approximation of a coagulation-Fragmentation Model for Animal Group Size Statistics\bigskip}

}
\bigskip

% Enter the first author's name and address:
\centerline{\scshape Pierre Degond}
\medskip
{\footnotesize
% please put the address of the first author
 \centerline{Department of Mathematics, Imperial College London,}
   \centerline{South Kensington Campus}
   \centerline{London SW7 2AZ, UK,}
   \smallskip
   \centerline{email: p.degond@imperial.ac.uk}
} % Do not forget to end the {\footnotesize by the sign }

\bigskip

\centerline{\scshape Maximilian Engel}
\medskip
{\footnotesize
 % please put the address of the second and third author
 \centerline{Department of Mathematics, Imperial College London,}
   \centerline{South Kensington Campus}
   \centerline{London SW7 2AZ, UK,}
   \smallskip
   \centerline{email: maximilian.engel13@imperial.ac.uk}
}

\bigskip

% date
\centerline{April 18, 2016}

\bigskip

\begin{abstract}
We study numerically a coagulation-fragmentation model derived by Niwa \cite{N} and further elaborated by Degond et al. \cite{DLP}. In \cite{DLP}  a unique equilibrium distribution of group sizes is shown to exist in both cases of continuous and discrete group size distributions. We provide a numerical investigation of these equilibria using three different methods to approximate the equilibrium: a recursive algorithm based on the work of Ma et. al. \cite{MJS}, a Newton method and the resolution of the time-dependent problem.
All three schemes are validated by showing that they approximate the predicted small and large size asymptotic behaviour of the equilibrium accurately. The recursive algorithm is used to investigate the transition from discrete to continuous size distributions and the time evolution scheme is exploited to show uniform convergence to equilibrium in time and to determine convergence rates. 
\bigskip

\noindent $\textit{Key words}$. Coagulation-fragmentation equation, numerics, convergence to equilibrium, fish schools, Newton method, Euler scheme

\bigskip
\noindent $\textit{AMS Subject Classification}$. 92D50, 92C31, 82B40, 82B44

\end{abstract}

\section{Introduction} \label{intro}
Most animals in nature aggregate in groups of different sizes. These sizes vary in their frequency and obviously depend on the species. So the question arises whether and how typical distributions of group sizes emerge. Related questions are: Can we find adequate models for these distributions? How do the distributions evolve over time? Is there an (or several) equilibrium distribution(s)? Can one say something about the trend towards these equilibria?\\
Various models of describing the coagulation and fragmentation of groups of animals have been suggested and analysed in the past (cf. e.g. \cite{BD, BDF, G, GL, O}). The model this work rests upon was introduced by Hiro-Sato Niwa in 2003 \cite{N} related to studies in \cite{N1, N2, N3} and has turned out to hold for data from pelagic fish and mammalian herbivores in the wild. The model can be formalized into coagulation-fragmentation integral equations where the coagulation rate is a constant independent from the group sizes and the fragmentation rate is also a constant independent from the fragment.
By analogy with an It\^{o} Stochastic Differential Equation Niwa shows that the equilibrium must be given by
\begin{equation} \label{Niwaequi}
W(N) \sim N^{-1} \exp \left[-\frac{N}{N_P}\left( 1 - \frac{e^{-N/N_P}}{2} \right) \right].
\end{equation}
$W(N)$ is the stationary probability density function of group sizes and $N_P$ is the average of the population distribution among group sizes, i.e. the expected size of the groups which an arbitrary individual is part of. For a continuum of cluster sizes, this is defined as
\begin{equation*}
N_P = \frac{\int N^2 W(N) dN}{\int N W(N) dN}.
\end{equation*}
In the discrete setting the integrals are replaced by sums.\\
In \cite{N} Niwa shows that the proposed equilibrium distribution (\ref{Niwaequi}) matches empirical data of several species of pelagic fish very well. Ma et al. \cite{MJS} provide a critical discussion of Niwa's result and point out some obscurities in the analysis.
Due to the appealing simplicity of Niwa's model and the good empirical match to the data, mathematical clarification is important.  Degond et al. \cite{DLP} have pursued with Niwa's model and given a rigorous description of the equilibria for continuous (model C) and discrete (model D) cluster sizes, which differ from (\ref{Niwaequi}). The lack of a detailed balanced condition has made the analysis difficult. However, by introducing the so called Bernstein transformation, they have shown that there exists a unique equilibrium, under a suitable normalization condition, for both the discrete and the continuous cluster size case.\\
The task of the present work is a numerical investigation of both models and their equilibria. 
The continuous equilibrium is approximated numerically using three different methods whose accuracy will be examined. One of them is a recursive algorithm derived from model D in \cite{MJS} which enables a transition from the discrete to the continuous equilibrium. The other two, a Newton and a time-dependent method, operate within a discretized truncated model, denoted by D', of the continuous model C. There is an abundant amount of literature about discretizations of coagulation (and fragmentation) integral equations using finite volume methods (e.g. \cite{BF, FL, FM, KK, KKW, VPMS}) or finite element methods (e.g. \cite{MR, NH, RJ}). In our case, the discretization scheme is already predetermined by model D.\\
It is investigated how well the numerically generated equilibria match the analytically predicted decay rate and the small-size asymptotic behaviour of the model C equilibrium. We find all three methods to be very accurate apart from small deviations of the large-size behaviour in the case of the Newton and the time-dependent method due to truncation. The Newton method turns out to be extremely fast, providing a very close approximation of the equilibrium after five iterations. The recursive algorithm is the best numerical approach to this particular model with respect to a couple of aspects: it is numerically cheap, doesn't require truncation, is completely accurate for the discrete model D and approximates the continuous case properly without any aberrations.
However, the other two methods are far more flexible regarding changes of the models since, in principal, they don't require constant coagulation-fragmentation parameters $p$ and $q$ as opposed to the recursive algorithm. The Newton scheme as an approach to prove the existence and uniqueness of the equilibrium, as introduced in this work for model C', has the advantage of not depending on fixed parameters as contrasted with the Bernstein method (see \cite{DLP}) which needs $p$ and $q$ to be equal to one. \\
Hence, the truncated model and the associated numerical methods provide the tools to work in more sophisticated models with the coagulation and fragmentation depending on the group sizes and/or time. In this context, the model under investigation serves as a toy model to show the accuracy of the suggested schemes. In addition to that, the Euler scheme helps to examine the convergence of time-dependent solutions to the stationary one, something that hasn't been understood comprehensively in the analysis in \cite{DLP}. The numerical approach indicates uniform convergence on finite intervals with super-exponential convergence rates independent from the group sizes.\\
We introduce model C and model D in Section~\ref{secequations}. In Section~\ref{secprelim} we summarise the analytical results concerning equilibria in models C and D. We introduce our own truncated model C' and the constructive approximation (Newton) method to its equilibrium in Section~\ref{sectruncation}. Section~\ref{secmethods} provides a description of the different numerical algorithms whose validations and insights are shown in Section~\ref{secinvestigations}.

\section{The governing equations: the continuous and the discrete version} \label{secequations}
\subsection{General form of the equations}
The continuous version of a coagulation-fragmentation equation, called also Smoluchowski equation, describes the evolution of the number density $f(x,t)$ of continuous sizes $x \geq 0$ at time $t$. In weak form it reads, for $\varphi$ being a test function:
\begin{equation} \label{weakgeneral}
\begin{split}
\frac{d}{dt}\int_{\mathbb{R}_{+}} \varphi(x)f(x,t)dx
= \frac{1}{2} \int_{(\mathbb{R}_{+})^2} (\varphi(x+y) - \varphi(x) - \varphi(y))a(x,y)f(x,t)f(y,t) dx dy \\
- \frac{1}{2} \int_{(\mathbb{R}_{+})^2}(\varphi(x+y) - \varphi(x) - \varphi(y))b(x,y) f(x+y,t) dx dy.
\end{split}
\end{equation}
The coagulation rate $a(x,y)$ and fragmentation rate $b(x,y)$ are both nonnegative and symmetric. The coagulation and fragmentation reactions can be written schematically
\begin{align*}
(x) + (y) \ &\xrightarrow{a(x,y)} \ (x+y) \quad \text{(binary coagulation)},\\
(x) + (y) \ &\xleftarrow{b(x,y)} \ (x+y) \quad \text{(binary fragmentation)}.
\end{align*}
By a change of variables, (\ref{weakgeneral}) can be transformed into 
\begin{equation} \label{weakgeneralalt}
\begin{split}
\frac{d}{dt}\int_{\mathbb{R}_{+}} \varphi(x)f(x,t)dx
= \frac{1}{2} \int_{(\mathbb{R}_{+})^2} (\varphi(x+y) - \varphi(x) - \varphi(y))a(x,y)f(x,t)f(y,t) dx dy \\
- \frac{1}{2} \int_{(\mathbb{R}_{+})}\left( \int_{0}^x(\varphi(x) - \varphi(y) - \varphi(x-y))b(y,x-y) dy \right) f(x,t) dx.
\end{split}
\end{equation}
Note that by taking $\varphi(x) = x$, one obtains the conservation of mass
\begin{equation} \label{mass}
\frac{d}{dt} \int_{\mathbb{R}_{+}} x f(x,t) dx = 0.
\end{equation}
The intuition behind  (\ref{weakgeneral}) becomes clearer when we consider the strong form. In the following, $Q_C$ shall denote the coagulation operator and $Q_F$ the fragmentation operator. They both have a gain and a loss component and build up the strong form of the equation as 
\begin{equation}\label{stronggeneral}
\frac{\partial f}{\partial t}(x,t) = Q_C (f) (x,t) + Q_F (f) (x,t),
\end{equation}
\begin{equation} \label{generalcoag}
Q_C(f)(x,t) = \frac{1}{2} \int_{0}^{x} a(y,x-y)f(y,t)f(x-y,t)  dy -  \int_{0}^{\infty} a(x,y)f(x,t)f(y,t)  dy,
\end{equation}
\begin{equation} \label{generalfrag}
Q_F(f)(x,t) = \int_{0}^{\infty} b(x,y)f(x+y,t) dy - \frac{1}{2}  \int_{0}^{x} b(y,x-y)f(x,t) dy.
\end{equation}
The case where the cluster sizes form a discrete set can be described analogously. So consider a system of clusters with discrete sizes $ i \in \mathbb{N}$. Merging and splitting with the coagulation rate $a_{i,j}$ and fragmentation rate $b_{i,j}$ are ruled by the following coagulation-fragmentation reactions
\begin{align*}
(i) + (j) \ &\xrightarrow{a_{i,j}} \ (i+j) \quad \text{(binary coagulation)},\\
(i) + (j) \ &\xleftarrow{b_{i,j}} \ (i+j) \quad \text{(binary fragmentation)}.
\end{align*}
The system is described by the number density $f_i(t)$ of clusters of size $i$ at time $t$ evolving according to the discrete coagulation-fragmentation equation. Written in weak form the equation reads for any test function $\varphi_i$
\begin{equation} \label{weakgeneralD}
\frac{d}{dt} \sum_{i=1}^{\infty} \varphi_i f_i(t) = \frac{1}{2} \sum_{i,j=1}^{\infty} ( \varphi_{i+j} - \varphi_i - \varphi_j)( a_{i,j} f_i(t) f_j(t) - b_{i,j} f_{i+j}(t)).
\end{equation}
The equation can also be written similarly to (\ref{weakgeneralalt}) as
\begin{equation} \label{weakgeneralaltD}
\begin{split}
\frac{d}{dt} \sum_{i=1}^{\infty} \varphi_i f_i(t) = \frac{1}{2} \sum_{i,j=1}^{\infty} ( \varphi_{i+j} - \varphi_i - \varphi_j) a_{i,j} f_i(t) f_j(t)\\
 - \frac{1}{2} \sum_{i=2}^{\infty} \left(\sum_{j=1}^{i-1} (\varphi_{i} - \varphi_j - \varphi_{i-j} )b_{j,i-j} \right) f_{i}(t).
\end{split}
\end{equation}
If one takes $\varphi_k = k$, it can be seen immediately that mass is conserved:
\begin{equation*}
\frac{d}{dt} \sum_{i=1}^{\infty} i f_i(t) = 0.
\end{equation*}
Let $Q_{Ci}$ and $Q_{Fi}$ denote the coagulation resp. fragmentation operator for cluster size $i$. Then the strong form can be written as
\begin{equation} \label{stronggeneralD}
\frac{\partial f_i}{\partial t} (t) = Q_{Ci}(f)(t) + Q_{Fi} (f) (t),
\end{equation}
\begin{equation} \label{coaggenerlD}
Q_{Ci}(f)(t) = \frac{1}{2} \sum_{j=1}^{i-1} a_{j,i-j} f_j (t) f_{i-j}(t) - \sum_{j=1}^{\infty} a_{i,j} f_i(t) f_j(t),
\end{equation}
\begin{equation} \label{fraggeneralD}
Q_{Fi}(f)(t) = \sum_{j=1}^{\infty} b_{i,j} f_{i+j}(t) - \frac{1}{2} \sum_{j=1}^{i-1} b_{j,i-j} f_i (t).
\end{equation}
\subsection{The equations based on Niwa's model}
According to Niwa's model, we assume $s$ different zones of space on which $\Phi$ individuals move. The number of individuals is conserved through time. At each time step every group moves towards a randomly selected site with equal probability. When $i$-and $j$-sized groups meet at the same site, they aggregate to a group of size $i+j$. So the coagulation rate is independent from the group sizes and can be written as $a_{i,j} = 2 q$ for any $i,j > 0$ where $q > 0$ is the fixed coagulation parameter. 
The fragmentation rate $b_{i,j}$  expresses the fact that at each time step each group with size $k \geq 2$ splits with probability $p$ independent of $k$, and that if it does split, it breaks into one of the pairs with sizes $ (1,k-1), (2,k-2), \dots, (k-1,1)$ with equal probability. As the actually distinct pairs are counted twice in such an enumeration, one gets for all $1 \leq i,j < k$ with $ i + j =k $: 
$b_{i,j} = \frac{p}{(i+j-1)/2} = \frac{2p}{i+j-1}$.
Summarizing, we can express Niwa's model in the discrete system of equations introduced above by choosing
\begin{equation} \label{Niwarates}
a_{i,j} = 2 q, \quad b_{i,j} = \frac{2p}{i+j-1}.
\end{equation}
As already indicated, Ma et al. \cite{MJS} have studied  the coagulation-fragmentation system with these rates. Gueron and Levin \cite{GL} had proposed coagulation and fragmentation rates that satisfied a detailed balance condition. That means that their choice of $a$ and $b$ was such that there exists an equilibrium distribution $\overline{f}$ fulfilling
\begin{equation*}
b(x,y) \overline{f}(x+y) = a(x,y) \overline{f}(x)  \overline{f}(y) \quad \forall x, y > 0.
\end{equation*}
The detailed balance condition is not satisfied in Niwa's model (cf. \cite[chapter 7]{DLP}). Degond et al. have chosen the same fragmentation and coagulation rates as Niwa in the continuous case but slightly different ones in the discrete case. The results of these steps are the discrete model D and the continuous model C, as described below:
\begin{itemize}
\item Model D (Discrete):
\begin{equation} \label{DegondrateD}
a_{i,j} = 2 q, \quad b_{i,j} = \frac{2p}{i+j+1}.
\end{equation}
\item Model C (Continuous):
\begin{equation} \label{DegondrateC}
a_{x,y} = 2 \overline{q}, \quad b_{x,y} = \frac{2\overline{p}}{x+y}.
\end{equation}
\end{itemize}
The fragmentation of a group of size $k$ in Model D can now be understood as breaking into the pairs $(0,i), \dots, (i,0) $ with equal probability $1/k+1$. This means that we also consider the cases in which actually nothing changes. This results in a significantly simpler analysis.\\\\
To summarize, we will consider the following models:
\subsubsection*{Model D}
The weak form for Model D (derived from (\ref{weakgeneralaltD})) reads, $\varphi_i$ being a test function,
\begin{equation} \label{weakNiwaD}
\begin{split}
\frac{d}{dt} \sum_{i=1}^{\infty} \varphi_i f_i(t) = q  \sum_{i,j=1}^{\infty} ( \varphi_{i+j} - \varphi_i - \varphi_j) f_i(t) f_j(t)\\
+ p \sum_{i=1}^{\infty} \left( -\varphi_{i} + \frac{2}{i+1} \sum_{j=1}^{i}\varphi_j  \right) f_{i}(t).
\end{split}
\end{equation}
The strong form becomes
\begin{equation} \label{strongNiwaD}
\frac{\partial f_i}{\partial t} (t) = Q_{Ci}(f)(t) + Q_{Fi} (f) (t),
\end{equation}
\begin{equation} \label{coagNiwaD}
Q_{Ci}(f)(t) = q \sum_{j=1}^{i-1}  f_j (t) f_{i-j}(t) - 2 q \sum_{j=1}^{\infty} f_i(t) f_j(t),
\end{equation}
\begin{equation} \label{fragNiwaD}
Q_{Fi}(f)(t) = -p f_i(t) + 2p\sum_{j=i}^{\infty} \frac{1}{j+1} f_{j}(t).
\end{equation}
\subsubsection*{Model C}
The continuous model C can be written in weak form, for any test function $\varphi$, as
\begin{equation} \label{weakNiwa}
\begin{split}
\frac{d}{dt}\int_{\mathbb{R}_{+}} \varphi(x)f(x,t)dx
= \overline{q} \int_{(\mathbb{R}_{+})^2} (\varphi(x+y) - \varphi(x) - \varphi(y)) f(x,t)f(y,t) dx dy \\
- \overline{p} \int_{(\mathbb{R}_{+})^2}(\varphi(x+y) - \varphi(x) - \varphi(y))\frac{f(x+y,t)}{x+y}  dx dy.
\end{split}
\end{equation}
or as 
\begin{equation} \label{weakNiwaalt}
\begin{split}
\frac{d}{dt}\int_{\mathbb{R}_{+}} \varphi(x)f(x,t)dx
= \overline{q} \int_{(\mathbb{R}_{+})^2} (\varphi(x+y) - \varphi(x) - \varphi(y))a(x,y)f(x,t)f(y,t) dx dy \\
+ \overline{p} \int_{(\mathbb{R}_{+})}\left( \frac{2}{x} \int_{0}^x \varphi(y) dy - \varphi(x) \right) f(x,t) dx .
\end{split}
\end{equation}
The strong form can be written as
\begin{equation} \label{strongNiwa}
\frac{\partial f}{\partial t}(x,t) = Q_C (f) (x,t) + Q_F (f) (x,t),
\end{equation}
\begin{equation} \label{coagNiwa}
Q_C(f)(x,t) = \overline{q} \int_{0}^{x} f(y,t)f(x-y,t)  dy - 2 \overline{q} \int_{0}^{\infty} f(x,t)f(y,t)  dy,
\end{equation}
\begin{equation} \label{fragNiwa}
Q_F(f)(x,t) = - \overline{p }f(x,t) + 2 \overline{p} \int_{x}^{\infty} \frac{f(y,t)}{y} dy .
\end{equation}
By introducing the method of Bernstein transformations, the existence and uniqueness of an equilibrium can be shown. The following section summarizes the important findings of \cite{DLP}, and prepares us for the numerical investigation.
\section{Preliminary findings in the analysis of the coagulation-fragmentation model from \cite{DLP}} \label{secprelim}
\subsection{Equilibrium in the continuous case}
Let $k \in \mathbb{N}$ and $f: x \in \mathbb{R}_+ \mapsto f(x) \in \mathbb{R}_{+}$. The $k$th moment $m_k(f)$ is given by
\begin{equation*}
m_k(f) = \int_{x \in \mathbb{R}_+} x^k f(x) dx.
\end{equation*}
For initial condition $f_0$ with $ m_1(f_0) < \infty $, we know from (\ref{mass}) that $m_1(f(t)) = m_1(f_0) := m_1 $.
There is a scaling invariance for model C:
\begin{proposition}
Let $f_0: x \in \mathbb{R}_{+} \mapsto f_0(x) \in \mathbb{R}_{+}$ be an initial condition for (\ref{strongNiwa}) with $m_1(f_0) =: m_1 < \infty$ and let $f_{\overline{p},\overline{q}}(x,t)$ be the solution of (\ref{strongNiwa}) with parameters $\overline{p}$ and $\overline{q}$. Then, on has
\begin{equation*}
f_{\overline{p},\overline{q}}(x,t) = \frac{\overline{p}^2}{m_1 \overline{q}^2} f_{1,1}(\frac{\overline{p}}{m_1 \overline{q}} x, \frac{\overline{p}^3}{m_1^2 \overline{q}^2}t),
\end{equation*}
where $f_{1,1}$ is associated with the initial condition $\tilde{f}_0$ such that
\begin{equation*}
\begin{split}
f_0 = \frac{\overline{p}^2}{m_1 \overline{q}^2} \tilde{f}_0(\frac{\overline{p}}{m_1 \overline{q}} x),\\
m_1(f_{1,1}(\cdot, t)) = m_1(\tilde{f}_0) = 1.
\end{split}
\end{equation*}
\end{proposition}
Due to this proposition, we can assume $\overline{p}=1$, $\overline{q}=1$ and $m_1=1$.
The problem in strong form becomes
\begin{equation} \label{strongscaled}
\frac{\partial f}{\partial t} (x,t) = \int_{0}^{x} f(y,t)f(x-y,t) dy - 2f(x,t)\int_{0}^{\infty} f(y,t) dy - f(x,t) + 2 \int_{x}^{\infty} \frac{f(y,t)}{y} dy,
\end{equation}
\begin{equation} \label{boundaryscaled}
m_1(f(\cdot,t)) = \int_{0}^{\infty} f(y,t) y dy = 1 \quad \forall t \in [0, \infty).
\end{equation}
In weak form it reads as
\begin{equation} \label{weakscaled}
\begin{split}
\frac{d}{dt} \int_{\mathbb{R}_{+}} \varphi(x) f(x,t) dx = \int_{\mathbb{R}_{+}^2} (\varphi(x+y) - \varphi(x) - \varphi(y))f(x,t)f(y,t) dx dy\\
 + \int_{\mathbb{R}_{+}}f(x,t) \left( \frac{2}{x} \int_{0}^x  \varphi(y)dy - \varphi(x) \right)  dx .
\end{split}
\end{equation}
\begin{definition}
A function $f: x \in (0, \infty) \to \mathbb{R} $ is said to be \textit{completely monotone} if it is $C^{\infty}$ and such that 
\begin{equation*}
(-1)^k f^{(k)} \geq 0, \quad \forall k \in \mathbb{N}.
\end{equation*}
\end{definition}
The main theorem can be stated as follows:
\begin{theorem} \label{equilibriumdetails}
There is a unique equilibrium distribution function $f_{\infty}$ of (\ref{strongscaled}) or (\ref{weakscaled}) satisfying (\ref{boundaryscaled}). It can be written as
\begin{equation*}
f_{\infty}(x) = \gamma (x) e^{-\frac{4}{27}x},
\end{equation*}
where $\gamma$ is a completely monotone function and has the following asymptotic behaviour:
\begin{equation} \label{smallsizegamma}
\gamma(x) \sim \frac{1}{3 \Gamma (4/3)}  x^{-2/3} 
 \quad \text{as} \ x \to 0,
\end{equation}
\begin{equation} \label{largesizegamma}
\gamma(x) \sim \frac{9}{16 \Gamma (3/2)} x^{-3/2} \quad \text{as} \ x \to \infty.
\end{equation}
\end{theorem}
\subsection{Equilibrium in the discrete case} \label{secprelimdisc}
One can show that there is a scaling invariance for model D as well (cf. \cite[Section 2.3]{DLP}). Hence, we will work with $p=q=1$ in the following.\\ 
In the discrete setting the $k$th moment of a sequence $f = (f_i)_{i \in \mathbb{N}}$ is given by
\begin{equation*}
m_k(f) = \sum_{i=1}^{\infty} i^k f_i.
\end{equation*}
Let us further introduce the sets
\begin{equation*}
\ell_{1,k} = \{f = (f_i)_{i \in \mathbb{N}}:\ f_i \geq 0, \ m_k(f) < \infty\}.
\end{equation*}
One can establish the well-posedness of the initial value problem (for a proof see \cite[Theorem 12.1]{DLP}):
\begin{theorem}
Let $k \geq 0$ and $ f_{in} =(f_{in,i})_{i \in \mathbb{N}}$ be given in $\ell_{1,k}$. Then there exists a unique global-in-time strong solution $ f \in C^1([0,\infty), l_{1,k})$ for system (\ref{strongNiwaD})-(\ref{fragNiwaD}) with $f(0) = f_{in}$.
If $k \geq 1$, then  $m_1(f(t)) = m_1(f_{in})$ for all $t\geq 0$.
\end{theorem}
Let $ F_t(x) $ denote a time-dependent solution of the continuous model C. For a transition from the continuous to the discrete model, we introduce a grid size $h > 0$ and the approximation
\begin{equation}\label{discretisation}
f_i^h \approx \int_{I_i^h} F_t(dx), \quad I_i^h := [ih, (i+1)h), \quad i =1,2,\dots 
\end{equation}
for the number of clusters with sizes in the interval $I_i^h$.\\
For a smooth test function $\varphi(x)$, we can write $\varphi_i = \varphi(ih)$ and require that $ f^h(t) =(f_i^h(t))_{i \in \mathbb{N}}$  solves Model D (as a discretization of Model C):
\begin{equation} \label{weakhD}
\begin{split}
\sum_{i=1}^{\infty}\varphi_i \frac{df_i^h}{dt}(t) = \sum_{i,j=1}^{\infty} (\varphi_{i+j}-\varphi_{i}-\varphi_{j}
)f_i^h(t) f_j^h(t) \\
+ \sum_{i=1}^{\infty} f_i^h(t) \left( -\varphi_i + \frac{2}{i+1} \sum_{j=1}^i \varphi_j \right).
\end{split}
\end{equation} 
Note that the genuine discrete case is given for $h=1$. Letting $h \to 0$, leads to an approximation of the continuous model by the discrete one.\\
We define the zeroth and first moment of an equilibrium distribution by
\begin{equation*}
m_0^h = \sum_{i=1}^{\infty} f_i^h, \quad m_1^h = \sum_{i=1}^{\infty} ih f_i^h.
\end{equation*}
The following theorem tells us that such an equilibrium actually exists and gives details about the asymptotic behaviour (cf. \cite{DLP}[Section 11 and 15]):
\begin{theorem} \label{largeasymptotics}
For any $m_1^h \in [0, \infty)$, there is a unique equilibrium solution $f^h =(f_i^h)_{i \in \mathbb{N}}$ of model D. The solution has the form
\begin{equation*}
f_n^h = \gamma_n z^{-n}, \quad z = 1 + \frac{4h}{27m_1^h},
\end{equation*}
where $\gamma$ is a completely monotone sequence with the asymptotic behaviour
\begin{equation*}
\gamma_n \sim \frac{9}{8} \left( \frac{m_1^h z}{h \pi} \right)^{1/2} n^{-3/2}  \quad \text{as} \ n \to \infty.
\end{equation*}
Further, the following mass-number relation holds:
\begin{equation} \label{massnumber}
\frac{m_0^h}{(1-m_0^h)^3} = \frac{m_1^h}{h}.
\end{equation}
\end{theorem}
Complete monotonicity in the discrete context means that
\begin{equation*}
(-1)^k (\Delta^k f^h)_n \geq 0 \quad \forall n, k \in \mathbb{N},
\end{equation*}
where the difference operator $\Delta$ is given by $ (\Delta f^h)_n = f_{n+1}^h - f_n^h$ and $(\Delta^k f^h)_n = \left( \Delta\left( \Delta^{k-1} f^h \right) \right) $.\\
Let $F_t^h$ denote the discrete measure on the grid $\{ih: \ i= 1, \dots \}$ formed from the solution $f^h(t)$ of model D:
\begin{equation} 
F_t^h(dx) = \sum_{i=1}^{\infty} f_i^h(t) \delta_{ih}(dx).
\end{equation}
Let again $ F_t(x) $ be a solution of the continuous model C. One can show that for a certain correspondence of initial data, for each $t>0$, we have $ F_t^h \to F_t $ narrowly as $ h \to 0$ (for a proof see \cite[Theorem 16.1]{DLP}).\\
In the following, we want to approximate these equilibria numerically. We are going to apply three different methods. The one based on model D will rest upon a recursive algorithm introduced in Section~\ref{secrec}. The other two, a Newton and a time-dependent method, require for a truncation in model C onto a compact interval of $\mathbb{R}$. This new model C' will be treated in the next section. 
\section{Model C': a truncated version of model C} \label{sectruncation}
\subsection{The time-dependent problem}
We introduce a truncation of the weak formulation of model C to the interval $[0,L]$. Let $\varphi$ be a test function. The truncation is chosen as follows:
\begin{definition}
A time-dependent size  distribution $f(t,x)$ in Model C' is characterised as a solution of the weak problem
\begin{equation} \label{weakformC'}
\begin{split}
\frac{d}{dt}\int_{0}^L \varphi(x) f(x,t)dx  = \int_{0\leq x+y \leq L} (\varphi(x+y) - \varphi(x) - \varphi(y))f(x,t)f(y,t) dx dy \\
 - \int_{0\leq x+y \leq L}(\varphi(x+y) - \varphi(x) - \varphi(y)) \frac{f(x+y,t)}{x+y} dx dy,
\end{split}
\end{equation}
for all $t > 0$ and test functions $\varphi$. 
\end{definition}
Note that, indeed, by chosing $ \varphi(x) =x $  mass conservation is still obtained:
\begin{equation*}
\frac{d}{dt} \int_{0}^L x f(x,t) dx = 0.
\end{equation*}
\begin{proposition}
Let $Q_{C_T}$ denote the coagulation operator and $Q_{F_T}$ the fragmentation operator. Then the strong form of model C' can be written down as:
\begin{equation} \label{strongformC'}
\frac{\partial f}{\partial t}(x,t) = Q_{C_T} (f) (x,t) + Q_{F_T} (f) (x,t),
\end{equation}
\begin{equation} \label{coagC'}
Q_{C_T}(f)(x,t) = \int_{0}^{x} f(y,t)f(x-y,t)  dy -  2 \int_{0}^{L-x} f(x,t)f(y,t) dy,
\end{equation}
\begin{equation} \label{fragC'}
Q_{F_T}(f)(x,t) = 2 \int_{x}^{L} \frac{f(y,t)}{y} dy - f(x,t).
\end{equation}
\end{proposition}
\begin{proof}
Obvious calculation.
\end{proof}
Further, we can state the following local existence and uniqueness result:
\begin{proposition}
Let $f_0 \in L^1([0,L])$ and $R > 0$. Then there is an $\alpha >0$ such that the initial value problem corresponding with (\ref{strongformC'})
\begin{equation*} 
\frac{\partial f}{\partial t}(x,t) = Q_{C_T} (f) (x,t) + Q_{F_T} (f) (x,t), \quad f(\cdot,0) = f_0(\cdot) \ \text{a.s.}
\end{equation*}
has a unique solution on $[0, \alpha]$ with values in $\overline{B}(f_0, R) \subset L^1([0,L])$.
\end{proposition}
\begin{proof}
This is an immediate application of the Cauchy-Lipschitz Theorem for initial value problems in Banach spaces as $Q_{C_T}$ is a continuous quadratic and $Q_{F_T}$ is a continuous linear operator from $L^1([0,L])$ to itself (cf. Lemma~\ref{boundedoperator}).
\end{proof}
\subsection{The equilibrium: a constructive approximation method} \label{secnewt}
We present a constructive approach to find the equilibrium in model C'. It relies on a Newton method. \\\\
The stationary version of (\ref{strongformC'}) is
\begin{equation*}
- Q_{F_T}(f)(x) = Q_{C_T}(f)(x).
\end{equation*}
This equation can also be written as
\begin{equation} \label{LinQuad}
Tf = q(f,f),
\end{equation}
with
\begin{align} 
Tf(x) &= f(x) - 2 \int_{x}^{L} \frac{f(y)}{y} dy = - Q_{F_T}(f)(x), \label{Lin} \\ 
q (f, \phi)(x) &= \int_{0}^{x} f(y) \phi(x-y) dy - \left( \int_{0}^{L-x} f(y) dy \right)\phi(x) - \left( \int_{0}^{L-x} \phi(y) dy \right)f(x). \label{Quad}
\end{align}
$T$ is a linear operator whereas $q$ is a bilinear form with $Q_{C_T}(f) = q(f,f)$.\\
Starting with an appropriate $f_0$, we want to find a recursive scheme giving a convergent sequence $(f^n)_{n \in \mathbb{N}}$ with limit $f_{\infty}$, the equilibrium. Observe the following: If $f^{n+1}$ was an equilibrium, we'd have
\begin{align*}
Tf^{n+1} &= Q_{C_T}(f^{n+1}) \\
&= Q_{C_T} \left(f^{n} + (f^{n+1} - f^n)\right) \\
&= Q_{C_T}(f^{n}) + 2q(f^n, f^{n+1} - f^n) + Q_{C_T}(f^{n+1} - f^n) \\
&= 2q(f^n, f^{n+1}) - Q_{C_T}(f^n)+ Q_{C_T}(f^{n+1} - f^n).
\end{align*}
with
\begin{equation*}
Q_{C_T}(f^{n+1} - f^n) = \mathcal{O}(f^{n+1} - f^n)^2
\end{equation*}
when $\left|f^{n+1} - f^n\right|$ is small.
Hence, the following Newton scheme rests upon neglecting this quadratic term and defines a sequence $(f^n)_{n \in \mathbb{N}}$ by iteratively solving the following linear problem:
\begin{equation} \label{NewtonScheme}
Tf^{n+1} - 2 q(f^n, f^{n+1} )  = - Q_{C_T}(f^n).
\end{equation}
Introducing $ \delta f = f^{n+1} - f^n$, by adding $ - T f^n$ and $ 2 Q_{C_T}(f^n)$ on both sides of equation (\ref{NewtonScheme}), we get
\begin{equation} \label{LinearProblem}
T \delta f - 2 q(f^n, \delta f) = -Tf^n + Q_{C_T}(f^n).
\end{equation}
We introduce the notation
\begin{equation*}
W_{f^n}(\delta f) = T \delta f - 2 q(f^n, \delta f), \quad G_n = -Tf^n + Q_{C_T}(f^n),
\end{equation*}
where $W_{f^n}$ is a linear operator and $G_n$ is a function. \\
$W_{\phi}$ can be written as
\begin{equation} \label{Woperator}
W_{\phi} f(x) = [ (1 + 2 \int_{0}^{L-x} \phi(y) dy )\Id - 2 K_{\phi}] f(x),
\end{equation}
where 
\begin{equation} \label{Koperator}
 K_{\phi}f(x) = \int_{x}^{L} \frac{f(y)}{y} dy + \int_{0}^{x} f(y) \phi(x-y) dy - \phi (x) ( \int_{0}^{L-x} f(y) dy).
\end{equation}
In the following, $\R{(W_{\phi})}$ denotes the range of $W_{\phi}$ and $\N{(W_{\phi})}$ its null space.\\ For $f \in L^1([0,L])$, $g \in L^{\infty}([0,L])$ we define 
\begin{equation*}
\langle f, g \rangle = \int_0^L f g  \ dx, 
\end{equation*}
and for $V \subset L^{\infty}([0,L])$ we define
\begin{equation*}
V^{\perp} = \{ f \in L^1([0,L]): \ \langle f, g \rangle = 0 \ \forall g \in V \}.
\end{equation*}

\begin{lemma}  \label{boundedoperator}
Let $\phi \in L^1([0,L])$. Then $W_{\phi}$, as given in equations (\ref{Woperator}), (\ref{Koperator}), is a bounded, liner operator from $L^1([0,L])$ to $L^1([0,L])$.
\end{lemma}
\begin{proof}
The Lemma follows immediately from the definitions.
\end{proof}
In addition, we can find out the following about the range of $W_{f^n}$ (We choose the index $f^n$ instead of $\phi$ in order to build on equation (\ref{LinearProblem})):
\begin{lemma} \label{range}
For any $f^n \in L^1 ([0,L])$ it holds that $\R{(W_{f^n})} \subset \spn{\{x\}}^{\perp} $. 
\end{lemma}
\begin{proof}
For any test function $\varphi$ we have 
\begin{align*}
&\int_{0}^L  [-Tf^n(x) + Q_{C_T}(f^n)(x)]\varphi(x) dx = \\
&=\int_{0\leq x+y \leq L} (\varphi(x+y) - \varphi(x) - \varphi(y))f^n(x)f^n(y) dx dy\\
&- \int_{0\leq x+y \leq L}(\varphi(x+y) - \varphi(x) - \varphi(y)) \frac{f^n(x+y)}{x+y} dx dy.
\end{align*}
So if we set $\varphi(x) = x$, we get
\begin{equation} \label{Gorthogonal}
0 = \int_{0}^L  [-Tf^n(x) + Q_{C_T}(f^n)(x)] x dx = \int_{0}^L  G_n(x) x dx.
\end{equation}
By adding and subtracting $Tf^{n}(x)$ and $ Q_{C_T}(f^{n})(x) $, one can see that 
\begin{align*}
0 &= \int_{0}^L  [-Tf^{n+1}(x) + Q_{C_T}(f^{n+1})(x)] x dx \\
  &= \int_{0}^L  [-Tf^{n}(x) + Q_{C_T}(f^{n})(x) - T \delta f (x) + Q_{C_T}(\delta f)(x) + 2 q(f^n, \delta f )(x) ] x dx.
\end{align*}
Since this is true for any $\delta f $ and the first two summands can be cancelled due to (\ref{Gorthogonal}), for any $\lambda > 0$ it holds that
\begin{equation*}
\int_{0}^L  [T( \lambda \delta f(x)) - 2 q(f^{n}, \lambda \delta f)(x)] x dx = \int_{0}^L  Q_{C_T}(\lambda \delta f)(x) x dx. 
\end{equation*}
Extracting the $\lambda$ and dividing by $\lambda$ leaves the factor $\lambda$ on the right hand side of the equation. Due to arbitrariness of $\lambda$, it can be chosen arbitrarily small which shows that the left hand side is zero.
\end{proof}
Now, we conjecture the following based on Fredholm theory (cf. \cite{B}):
\begin{conjecture} \label{conjecture}
 $\R{(W_{\phi})} = \spn{\{x\}}^{\perp}$, $ \dim{ \N{(W_{\phi})}} = 1$ and $ \N{(W_{\phi})} \cap \spn{\{x\}}^{\perp} = \{0\}$.
\end{conjecture}
Proving this conjecture allows to single out the solution of $W_{\phi}f = g$ by imposing $\int x f dx =1$. This is the subject of current work. 
\section{Numerical methods} \label{secmethods}
This section contains three numerical methods to approach an equilibrium distribution. The first one concerns a recursive computation of the equilibrium sequence for model D already proposed in \cite{MJS} and \cite{DLP}. The other approaches rely on model D', a discretised version of truncated model C'. The first one simulates the evolution of the size distribution in time via an explicit Euler scheme and shall reach the steady state after a certain time span. The other one follows the Newton method theoretically outlined in Section~\ref{secnewt}. Note that the second method provides also an approximation of the time-dependent problem while the first and third methods only allow for the computation of the equilibrium.
\subsection{A recursive algorithm for model D equilibria} \label{secrec}
The equilibrium sequence in model D, $(f_i^h)_{ i \in \mathbb{N}}$, can be computed recursively for any $h > 0 $ (see \cite[Section 4.2.3]{DLP} and \cite[Eq. (13)-(15)]{MJS}).\\
For a test function $\varphi$ with $ \varphi_i = \varphi(ih)$, the equilibrium profile satisfies
\begin{equation*}
0 = \sum_{i,j =1}^{\infty} [ \varphi_{i+j} - \varphi_{i} - \varphi_j]f_i^h f_j^h + \sum_{i = 1}^{\infty} f_i^h [\frac{2}{i+1} \sum_{j=1}^i \varphi_j - \varphi_i].
\end{equation*}
Define now
\begin{equation*}
m_0^h = \sum_{j=1}^{\infty} f_j^h, \quad b_i = \sum_{j=i}^{\infty} \frac{1}{j+1} f_j^h.
\end{equation*}
Taking $\varphi_j \equiv 1$ yields
\begin{equation} \label{m0formula}
0 = -(m_0^h)^2 -m_0^h + 2 \sum_{i=1}^{\infty} \frac{i}{i+1} f_i^h = -(m_0^h)^2 + m_0^h - 2b_1.
\end{equation}
Further, with taking $\varphi_k = 1$ if $k=i$ and $0$ otherwise, we get
\begin{equation*}
0 =  \sum_{j=1}^{i-1}  f_j^h f_{i-j}^h - (2m_0^h + 1)f_i^h + 2b_i, \quad i \geq 1. 
\end{equation*}
Based on these equations, one gets the following recursive algorithm: \\
Choose $m_0^h \in (0,1)$ and set 
\begin{equation} \label{startofsequence}
b_1 = \frac{1}{2}(-(m_0^h)^2 + m_0^h).
\end{equation}
Then for $ i = 1,2,3, \dots $: 
\begin{align}
f_i^h &= (1+2m_0^h)^{-1} \left( 2b_i + \sum_{j=1}^{i-1} f_j^h f_{i-j}^h \right), \label{fsequence}\\
b_{i+1} &= b_i - \frac{f_i^h}{i+1}. \label{bsequence}
\end{align}
\subsection{Model D': the discretized form of model C'}
\subsubsection{Setting of the model}
We consider solutions $f(x,t)$ of the truncated model C' and write $f_i(t) = f(ih,t)$ for the discretised function. Let $L >0$ be the truncation size, $h$ the grid size and $N= L/h$.  Write $\phi_i = \phi(ih)$ for a test function $\phi$.
\begin{definition}
The weak form of model D', the discretisation of model C', is given by the following evolution equation for the discrete size distribution $f_i(t)$:
\begin{align*}
\frac{d}{dt} \sum_{i=1}^N h f_i(t) \phi_i &= \sum_{2 \leq i+j \leq N} h^2 [ \phi_{i+j} - \phi_{i} - \phi_j]f_i(t) f_j(t) \\
& + \sum_{1\leq i \leq N} h f_i(t) [\frac{2}{i+1} \sum_{j=1}^i \phi_j - \phi_i],
\end{align*}
for all test sequences $\phi_i$.
\end{definition}
Observe that mass is preserved over time according to this equation. 
\begin{remark}
Note that the link between model C' and model D' resembles the link between model C and D as discussed in Section~\ref{secprelimdisc}. However, note that equation (\ref{discretisation}) defines $f_i^h(t)$ to be interpreted as $hF_t(ih)$, if $F_t(x)$ is a solution of model C.
\end{remark}
\begin{proposition}
The strong form of model D' is given by
\begin{equation} \label{discretestrong}
\frac{d}{dt}f_i(t) = h\sum_{j=1}^{i-1} f_{i-j}(t) f_j(t) -2h f_{i}(t) \sum_{j=1}^{N-i} f_j(t)
 - f_i(t) + 2 \sum_{j=i}^N \frac{f_{j}(t)}{j+1}.
\end{equation}
for $1\leq i \leq N$.
\end{proposition} 
\begin{proof}
Obvious calculation.
\end{proof}

\subsubsection{Time discretization of the time-evolution scheme} \label{sectimemethod}
The explicit Euler scheme in time is applied with time step size $\Delta t$. Let $ t_k = k \Delta t $. The sequence $ \{f^k\}_{k \in \mathbb{N}} $ denotes  an approximation of $ \{f(t_k)\}_{k \in \mathbb{N}} $  and is defined by the following recursive scheme:
\begin{equation} \label{dynamicalscheme}
f^{k+1} = f^{k} + (\frac{df}{d\tau} )^k\Delta t,
\end{equation}
where for any point $ih$ with $ 1 \leq i \leq N $,  $ \left\{ \left(\frac{df_i}{d\tau}\right) ^k \right\}_{i=1, \dots, N}$ is given by 
\begin{equation} \label{Eulerfactor}
\left(\frac{df_i}{d\tau} \right)^k  = h\sum_{j=1}^{i-1} f_{i-j}^{k} f_{j}^k  -2h f_{i}^{k} \sum_{j=1}^{N-i} f_{j}^{k}
 - f_{i}^{k} + 2 \sum_{j=i}^N \frac{f_{j}^{k}}{j+1}.
\end{equation}
The time-step is adjusted recursively. Starting with $dt = 10^{-2}$, the time step size is increased by ten per cent as long as the distribution stays non-negative and monotone. If one of these criteria is violated, the step size is reduced by ten per cent. The maximal time step size given by that scheme is $dt = 1.1$.
\subsubsection{Equilibrium in model D': the Newton method} \label{secnewtmethod}
The stationary equation in the discretized setup of model C' reads,
for $1\leq i \leq N$,
\begin{equation*}
0 = h\sum_{j=1}^{i-1} f_{i-j} f_j -2h f_{i} \sum_{j=1}^{N-i} f_j
 - f_i + 2 \sum_{j=i}^N \frac{f_{j}}{j+1}.
\end{equation*}
Analogously to Eq. (\ref{LinQuad}) involving the operators $T$ and $q$, the discretized problem can be written as
\begin{equation} \label{Sp}
Sf = p(f,f),
\end{equation}
where for $ 1 \leq i \leq N$
\begin{align*}
(Sf)_i &= f_i - 2 \sum_{j=i}^{N} \frac{f_j}{j+1}  , \\
(p (f , g))_i &= \sum_{j=1}^{i-1} h f_j g_{i-j} - f_i \sum_{j=1}^{N-i} h  g_j - g_i \sum_{j=1}^{N-i} h  f_j.
\end{align*}
$S$ is a linear operator and $p$ is a bilinear form. Write $P(f) = p(f,f)$. Hence, the task is to find $f$ such that its image under the linear operator $S$ equals its image under the quadratic form $P$ derived from the bilinear form $p$.\\
Following our considerations in Section~\ref{secnewt}, we apply the Newton method expressed by Eq. (\ref{NewtonScheme}). Starting with an appropriate $f^0$ the following recursive scheme is applied:
\begin{equation*}
Sf^{n+1} - 2p(f^{n+1},f^n) + P(f^n) = 0.
\end{equation*}
The limit of this sequence, if it exists, satisfies the stationary equation (\ref{Sp}). \\
Analogously to (\ref{LinearProblem}), the recursive scheme can be written as
\begin{equation*}
S \delta f - 2p(f^n, \delta f) = -Sf^n + P(f^n, f^n),
\end{equation*}
where we introduce the notation
\begin{equation*}
V_{f^n}(\delta f) = S \delta f - 2p(f^n, \delta f), \quad H_n = -Sf^n + P(f^n, f^n).
\end{equation*}
This equation can be written explicitly as 
\begin{align}
(H_n)_i &=  -2 \left( \sum_{j=i}^{N} \frac{(\delta f)_{j}}{j+1} +   \sum_{j=1}^{i-1} h (\delta f)_{j} f_{i-j}^n - f_{i}^{n} \sum_{j=1}^{N-i} h  (\delta f)_{j} \right) \nonumber \\
&+ \left( 1 + 2  \sum_{j=1}^{N-i} h f_{j}^{n} \right) (\delta f)_{i}.
\end{align}
We transfer our considerations concerning the invertibility of $W_{f^n}$ in Section~\ref{secnewt} to the discretised version $V_{f^n}$. Let $x = (1, \dots, N)$. The range of the operator is restricted to $ \spn \{x\}^{\perp} $, i.e. to $N-1$ dimensions, and, hence, consider the above equation just for $ 1 \leq i \leq N-1$. Thereby we win a degree of freedom to implement the mass conservation in form of
\begin{equation*}
 (\delta f)_N = \left(-\sum_{i=1}^{N-1} i (\delta f)_i \right)/N.
\end{equation*}
This scheme provides us with an algorithm to approximate numerically the solution of the stationary problem (\ref{LinQuad}). As always, the performance of Newton's method crucially depends on the choice of the initialization. Here, we choose
\begin{equation} \label{exponential}
 f_{i}^{0} = \frac{m_1 \exp(-ih)}{h \sum_{j=1}^N jh \exp(-jh)},
 \end{equation}
with $m_1 >0$ denoting the mass to be chosen which will lead to convergence. 
\section{Numerical investigations} \label{secinvestigations}
The numerical methods introduced in Section~\ref{secmethods} shall now be applied. In the first subsection we check if the computed equilibrium distributions actually show the behaviour analytically predicted in \cite{DLP}. Hence, we have to account for non-negativity and the predicted asymptotics for small and large sizes. We supplement the validation of the schemes by a comparison of the large-size asymptotics in model D and model C. Further, we exploit the codes to gain new insights into the small-size behaviour in model D and the convergence rates to equilibrium in time. In the following, it will be appropriate to display the distributions mainly in a log scale using the decadic logarithm if not declared otherwise.
\subsection{Validation of the numerical schemes}
\subsubsection{The Newton method}
First, we want to check the accuracy of the Newton method presented in the previous section. In particular, we will compare the predicted asymptotic behaviour with the asymptotic behaviour displayed by the computed equilibrium distribution.
Recall from Theorem \ref{equilibriumdetails} that according to equation (\ref{smallsizegamma}) the unique equilibrium $f_{\infty}$ for mass $m_1 = 1$ satisfies 
\begin{equation}
\log_{10}{f_{\infty}(x)} \sim \log_{10}{\frac{1}{\Gamma (4/3)}} - \frac{4}{27}x\log_{10}{e}- (2/3)\log_{10}{x} \quad \text{as} \ x \to 0. \label{fsmalllog}
\end{equation}
Due to equation (\ref{largesizegamma}), the large-size asymptotic behaviour of $f_{\infty}$ is given by
\begin{equation} 
\log_{10}{f_{\infty}(x)} \sim \log_{10}{\frac{9}{16 \Gamma (3/2)}} - \frac{4}{27}x\log_{10}{e}- (3/2)\log_{10}{x} \quad \text{as} \ x \to \infty. \label{flargelog}
\end{equation}
The following plots show that the approximation of the equilibrium generated by the Newton method  matches the predicted asymptotic behaviour very well.
First, we are interested in the asymptotic behaviour for large sizes. We choose $m_1 = 1$, truncation size $L=100$ and $h=0.01$. We perform five iterations.
In Fig.~\ref{fig_5_1}, the solid blue line shows the logarithmic distribution as a function of the group sizes whereas the dashed red line shows the predicted asymptotic behaviour for $x \to \infty$ in (\ref{flargelog}).
The distribution is chosen in a log scale while the group size is shown in a linear scale in order to illustrate the leading behaviour for the logarithmic distribution, $-\frac{4}{27}x \log_{10}{e} $, in a linear shape.
Second, we focus on the small-size behaviour $x \to 0$. We truncate at $L=5$ and take $h=0.0005$. Since for the case of $m_1 =1$ and a calculation up to $L=100$, the mass concentrated in $[0,5] $ equals $0.5676$, we take this as our starting value for the mass. Again, we perform five iterations. In Fig.~\ref{fig_5_2} the blue solid graph shows the log of the distribution as a function of the log of the group sizes whereas the red dashed graph shows the predicted asymptotic behaviour close to 0. These graphs show a linear behaviour consistent  with the leading order term being given by $-(2/3) \log_{10}{x} $ (see Eq. (\ref{fsmalllog})).\\\\
Note that the distribution as shown in Fig.~\ref{fig_5_1} tends to zero very quickly (already $ f_{\infty}(x) < 10^{-2}$ at $x=10$) but never becomes negative as intended.
Observe the perfect convergence of both graphs for the group sizes becoming higher and higher. This means that the large-size asymptotic behaviour of the equilibrium generated by the Newton scheme is utterly accurate.
\begin{figure}[ht]
\centering
\begin{subfigure}{.5\textwidth}
  \centering
  \includegraphics[width=1\linewidth]{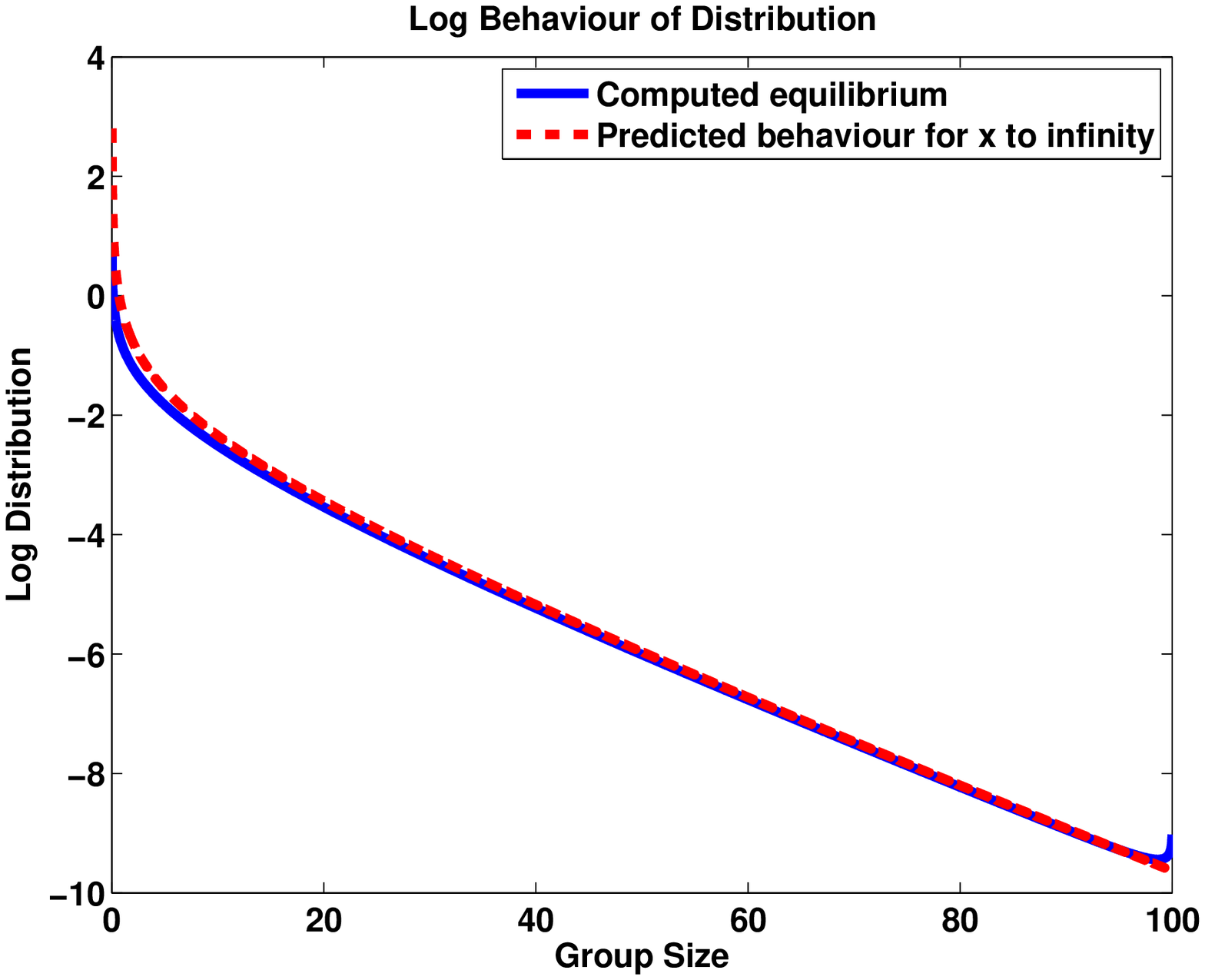}
  \caption{Illustration of large-size behaviour}
  \label{fig_5_1}
\end{subfigure}%
\begin{subfigure}{.5\textwidth}
  \centering
  \includegraphics[width=1\linewidth]{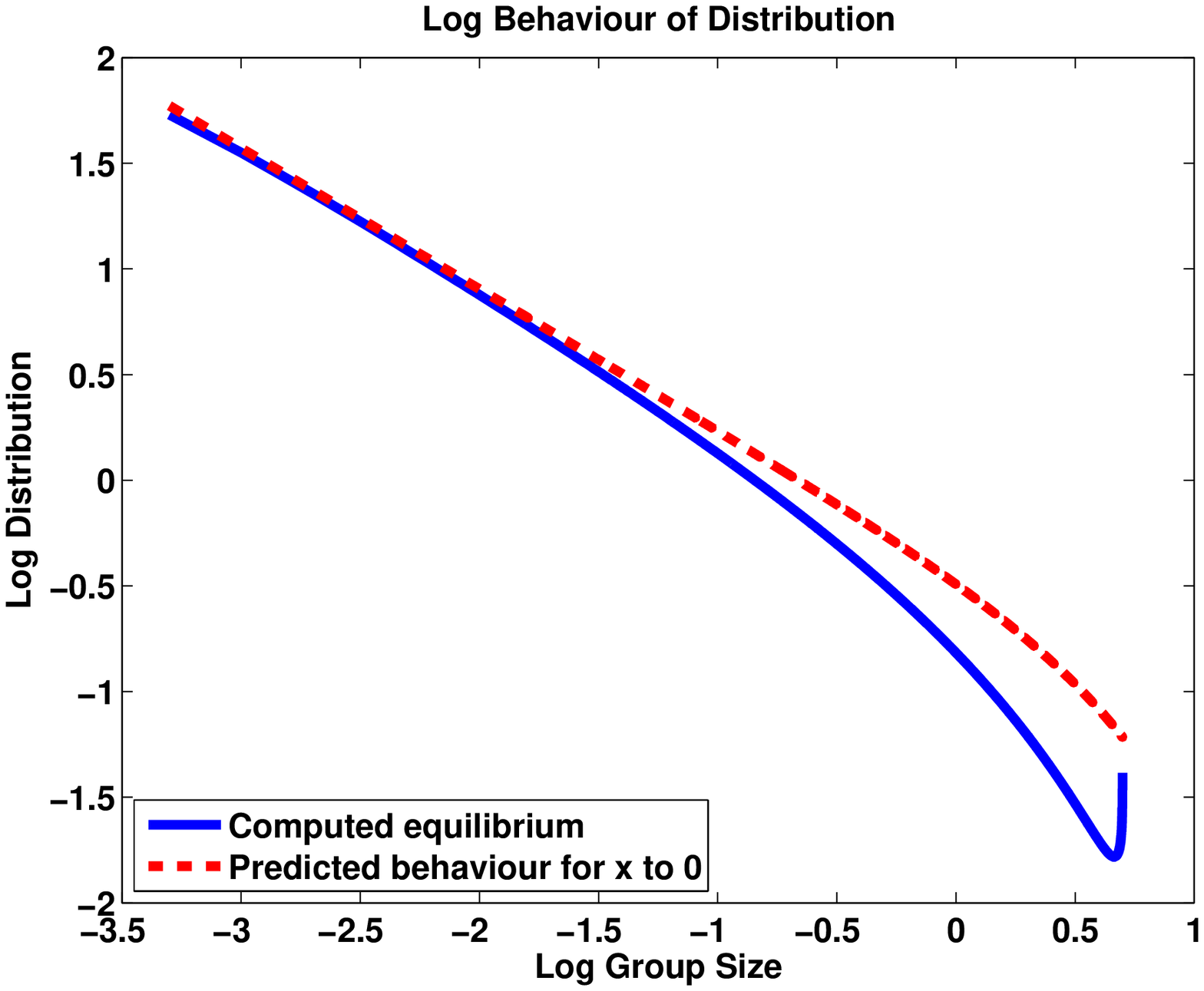}
  \caption{Illustration of small-size behaviour}
  \label{fig_5_2}
\end{subfigure}
\caption{The equilibrium distribution is approximated by the Newton scheme (Section~\ref{secnewtmethod}). In Fig.~\ref{fig_5_1}, we take mass $m_1 =1$, grid size $h=0.01$ and the cut-off at $L=100$. The plot shows the generated distribution (blue solid line) in a log scale against the group sizes in a linear scale and presents the theoretically found large-size asymptotic behaviour (red dashed line) in a log scale for the sake of comparison. The group sizes are taken in a linear scale in order to illustrate the leading behaviour for large group sizes as a straight line.
For Fig.~\ref{fig_5_2}, the equilibrium distribution is approximated by the Newton scheme  taking mass $m_1 =1$, grid size $h=0.0005$ and the cut-off at $L=5$. The plot shows the generated distribution (blue solid line) in a log scale and presents the theoretically found asymptotic small-size behaviour (red dashed line) in a log scale for the sake of comparison. The group sizes are taken in a log scale as well in order to illustrate the leading behaviour close to zero as a straight line.}
\label{fig:Newton_val}
\end{figure}
There is a very small kink at the cut-off at $L=100$.
This is a consequence of the truncation. In  model C' the groups of size $L=100$ cannot be part of coagulation into a group of bigger size, as opposed to model C which is defined on $[0, \infty)$. Also groups with sizes slightly smaller than $100$ are concerned as they are involved in significantly less coagulation than in the case without truncation. Summarizing, the cut-off leads to a small overestimate of the probability of occurrence for group sizes in a small neighbourhood of $100$ compared to model C.
Varying $h$ in the range $(0, 0.1)$ doesn't make a visible difference regarding the kink. For $ 0.1 \leq h \leq 1$ the kink becomes much smaller. This indicates that the missing coagulation concerns mainly a neighbourhood of $L$ with radius $0.1$. Group sizes outside that range are not visibly affected by not being able to merge into groups of size bigger than $100$.
\\\\
Note the approach of both graphs for $x \to 0$ in Fig.~\ref{fig_5_2}.
We can see a high similarity to the predicted small-size behaviour but no real convergence. This divergence close to 0 can be explained by the fact that model C is continuous and has a singularity at 0 whereas the numerical equilibrium is discrete. Further recall from equations (\ref{DegondrateD}) and (\ref{DegondrateC}) that we have chosen the discrete fragmentation rate to be $b_{i,j} = \frac{2}{i+j+1}$ whereas the continuous rate is given by $b_{x,y} = \frac{2}{x+y}$. Hence, the fragmentation probability is smaller in the discrete setting than in model C. This explains that the generated distribution lies beneath the asymptotic behaviour of model C.\\\\
If we choose the computed equilibrium distributions shown in Fig.~\ref{fig:Newton_val} as initial distributions for the  time-dependent scheme described in Section~\ref{sectimemethod}, they actually stay the same over an arbitrary long period of time (taking time step size $dt \leq 1.1$). This confirms that the computed equilibrium is indeed a proper approximation of the stationary solution of (\ref{strongformC'})-(\ref{fragC'}).
\subsubsection{The Euler scheme}
Let us now turn to the convergence to the equilibrium in the time evolution scheme. In the following we start with a uniform distribution. We take the time step size $dt = 1$ (which is accurate due to the remark in Section~\ref{sectimemethod}) and work with $m_1 =1$.
We observe in Fig.~\ref{fig:Dynamic} that there is actually convergence to the equilibrium. Again, start with the large sizes and take the truncation size $L=100$ and the grid size $h=0.01$. The stationary distribution reached after time length $T=30$ has exactly the same shape as Fig.~\ref{fig_5_1}. As we can see in Fig.~\ref{fig_5_3}, the predicted large-size asymptotics are reached. As in the case of the Newton algorithm, one can also observe the kink at the cut-off due to the reason explained above. For the investigation of the small-size behaviour, we truncate at $L=5$ and take $h=0.0005$. As in the case of the Newton algorithm for generating the equilibrium, we choose $0.5676$ as starting value for the mass to simulate the process for an overall mass of $m_1 =1$. For generating the small-size behaviour accurately enough, we have to choose $dt = 0.5$. 
After $T=6$ we get the small-size behaviour displayed in the following Fig.~\ref{fig_5_4}. It seems to equal the predicted asymptotics up to a point very close to 0 where it diverges slightly from the theoretical prediction. This is exactly the same observation as in the Newton scheme. The possible reasons are obviously the same.
\begin{figure}[ht]
\centering
\begin{subfigure}{.5\textwidth}
  \centering
  \includegraphics[width=1\linewidth]{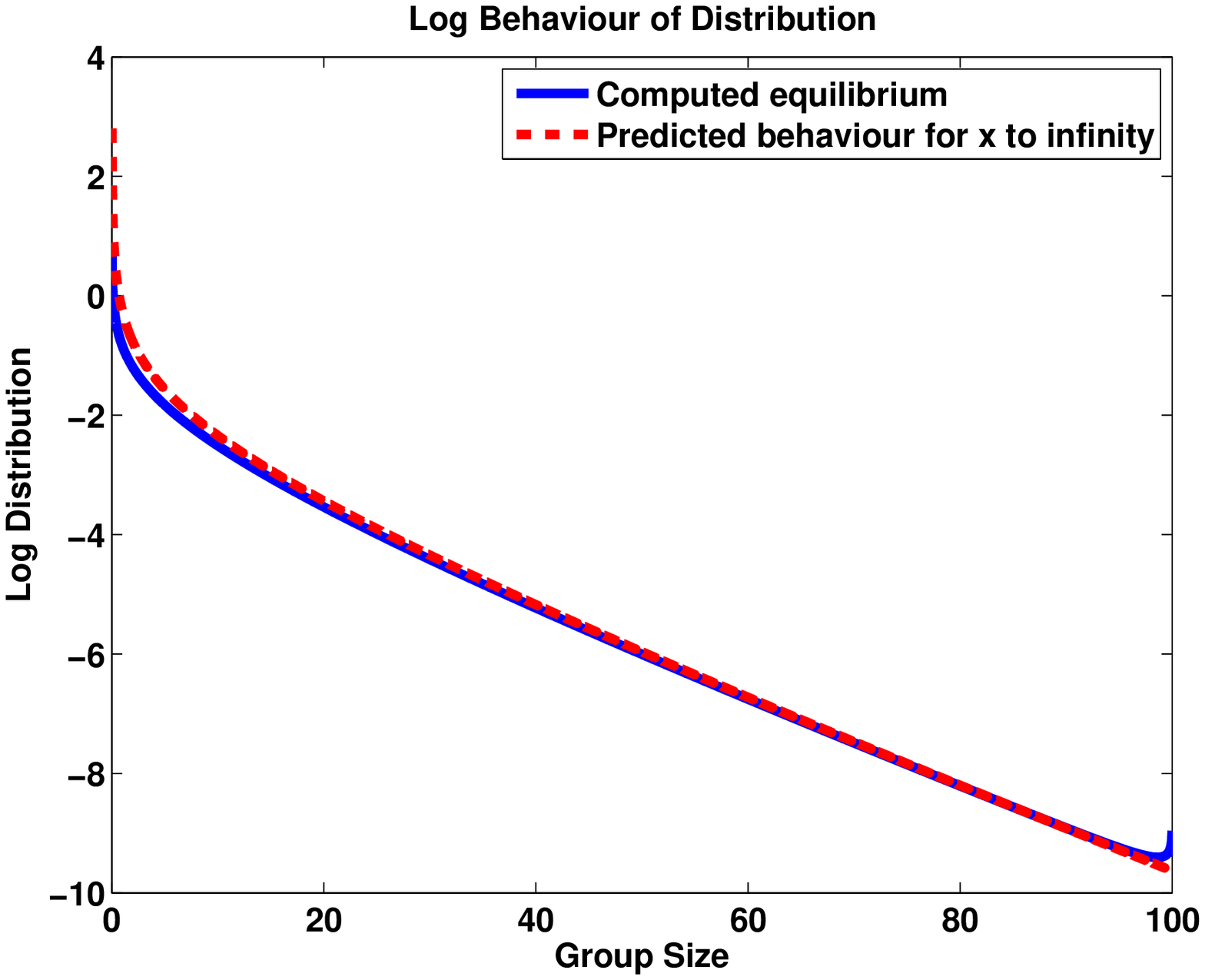}
  \caption{Illustration of large-size behaviour}
  \label{fig_5_3}
\end{subfigure}%
\begin{subfigure}{.5\textwidth}
  \centering
  \includegraphics[width=1\linewidth]{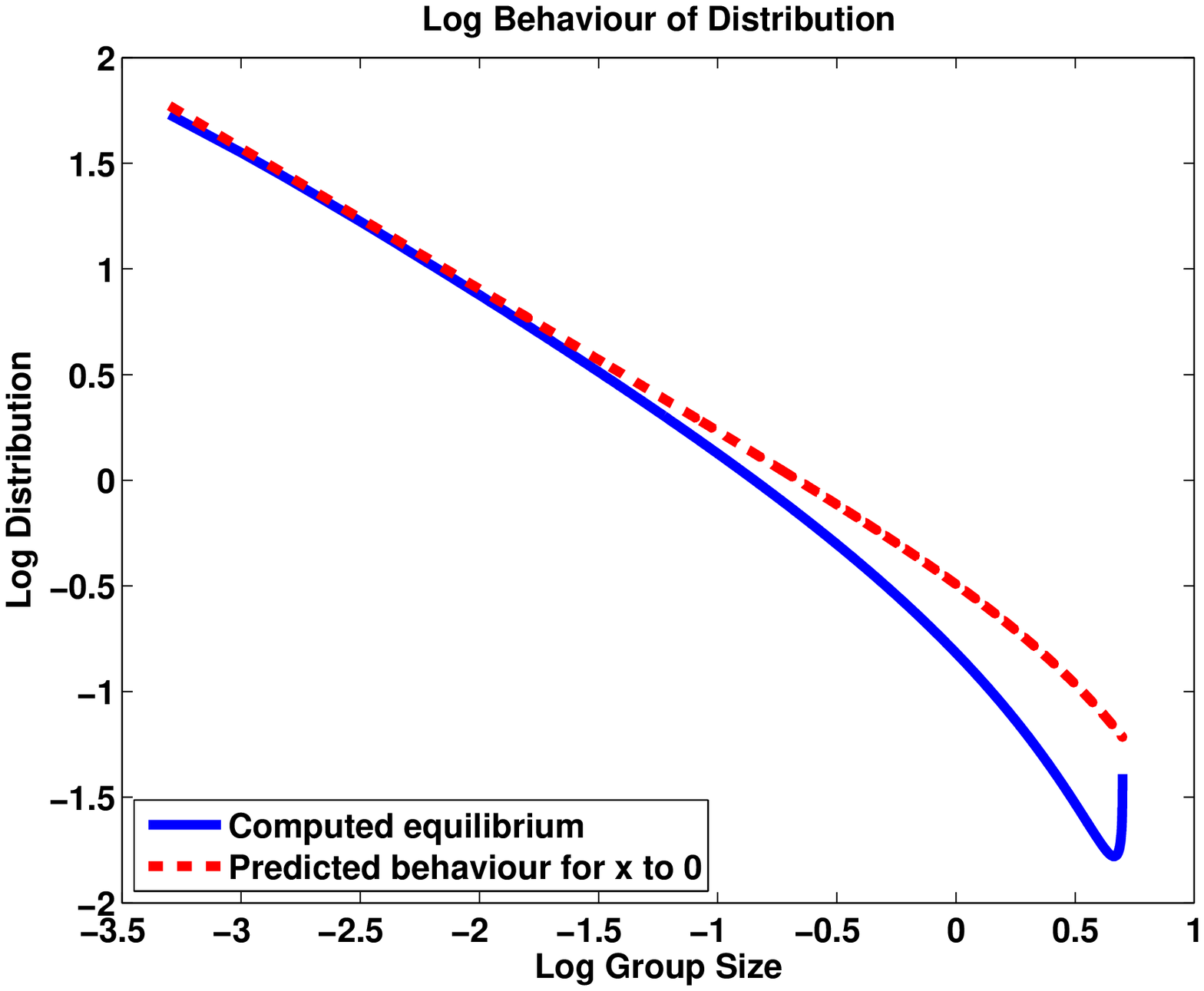}
  \caption{Illustration of small-size behaviour}
  \label{fig_5_4}
\end{subfigure}
\caption{The equilibrium distribution is approximated by simulating the time evolution of the distribution via the Euler scheme. Starting with a uniform distribution, the equilibrium, is reached at $T=30$ at the latest. In Fig.~\ref{fig_5_3}, we take mass $m_1 =1$, grid size $h=0.01$ and the cut-off at $L=100$. The plot shows the generated distribution (blue solid line) in a log scale as a function of the group sizes in a linear scale and presents the theoretically found large-size asymptotic behaviour (red dashed line) in a log scale for the sake of comparison. The group sizes are taken in a linear scale in order to illustrate the leading behaviour for large group sizes as a straight line.
For Fig.~\ref{fig_5_4}, the equilibrium distribution is approximated by the Euler scheme taking mass $m_1 =1$, grid size $h=0.0005$ and the cut-off at $L=5$. The plot shows the generated distribution (blue solid line) in a log scale and presents the theoretically found asymptotic small-size behaviour (red dashed line) in a log scale for the sake of comparison. The group sizes are taken in a log scale as well in order to illustrate the leading behaviour close to zero as a straight line.}
\label{fig:Dynamic}
\end{figure}
\subsubsection{The recursive computation of the equilibrium sequence}
Now we turn to checking the accuracy of the recursive scheme introduced in Section~\ref{secrec}. In the following we will choose $m_1^h$ and then $m_0^h$ such that equation (\ref{massnumber}) is satisfied. Using the recursive algorithm determined by equations (\ref{startofsequence})-(\ref{bsequence}), one can compute the equilibrium $(f_i^h)_{i \in \mathbb{N}}$ up to an arbitrarily large integer. As opposed to model C', we do not have to care about truncation. For the sake of comparison with the continuous model, we will look at $f_i^h$ as $h f(ih)$ in accordance with equation (\ref{discretisation}).\\\\
Again, we want to compare the predicted asymptotic behaviour with the asymptotic behaviour displayed by the computed equilibrium distribution:
recall from Theorem (\ref{largeasymptotics}) that the equilibrium $f_n^h$ for mass $m_1^h$ satisfies the large-size asymptotic behaviour given by
\begin{equation}
\log_{10}{f_n^h} \sim \log_{10}{C} - n\log_{10}{z}- (3/2)\log_{10}{n} \quad \text{as} \ n \to \infty \label{reclargelog},
\end{equation}
where
\begin{equation*}
z = 1 + \frac{4h}{27 m_1^h} ,\quad C = (9/8)\sqrt{\frac{m_1^h z}{h \pi}}.
\end{equation*}
There is no theoretical prediction for the small-size behaviour since the recursive scheme was derived from the discrete model which obviously doesn't have an equilibrium with singularity at zero as opposed to the continuous case. However, we will discuss the possibility of a small-size analysis in Section~\ref{secsmallsize}.\\\\
The plots in Fig.~\ref{fig:rec} indicate that the distribution generated by the algorithm matches very well the predicted asymptotic behaviour for the equilibrium for any $h >0$. Again for the sake of comparison with the continuous setting, we choose $m_1^h =1$ and compute the terms of the sequence until $L=100$. 
In Fig.~\ref{fig_5_5}, we choose the grid size $h =1$ which gives the actual realistic distribution with integer group sizes. The  plot compares the predicted asymptotic behaviour given by Eq.(\ref{reclargelog}) with the one given by our computed equilibrium. In Fig.~\ref{fig_5_6}, we do the same for $h =0.01$.
\begin{figure}[htp]
\centering
\begin{subfigure}{.5\textwidth}
  \centering
  \includegraphics[width=1\linewidth]{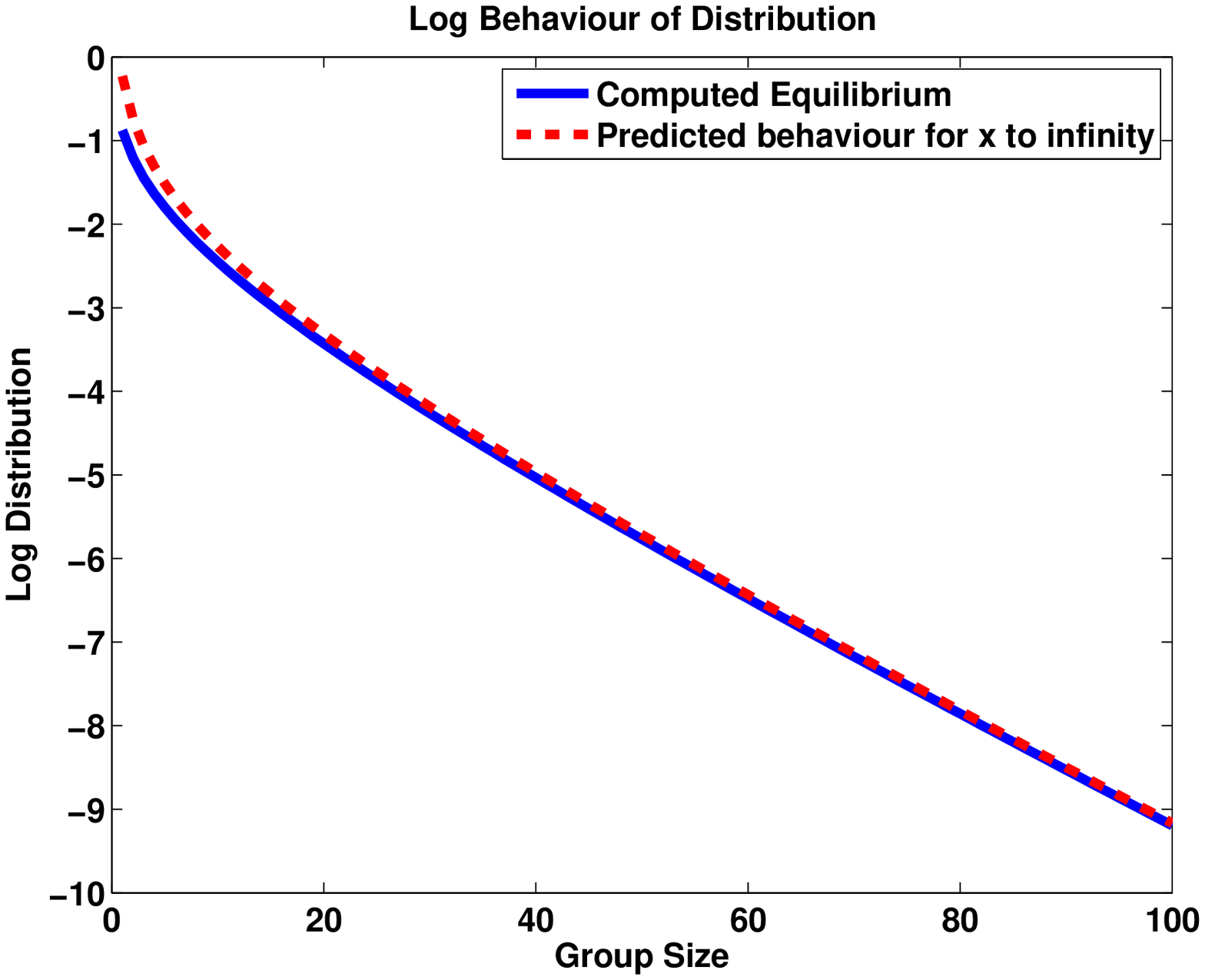}
  \caption{Log Distribution for $h=1$}
  \label{fig_5_5}
\end{subfigure}%
\begin{subfigure}{.5\textwidth}
  \centering
  \includegraphics[width=1\linewidth]{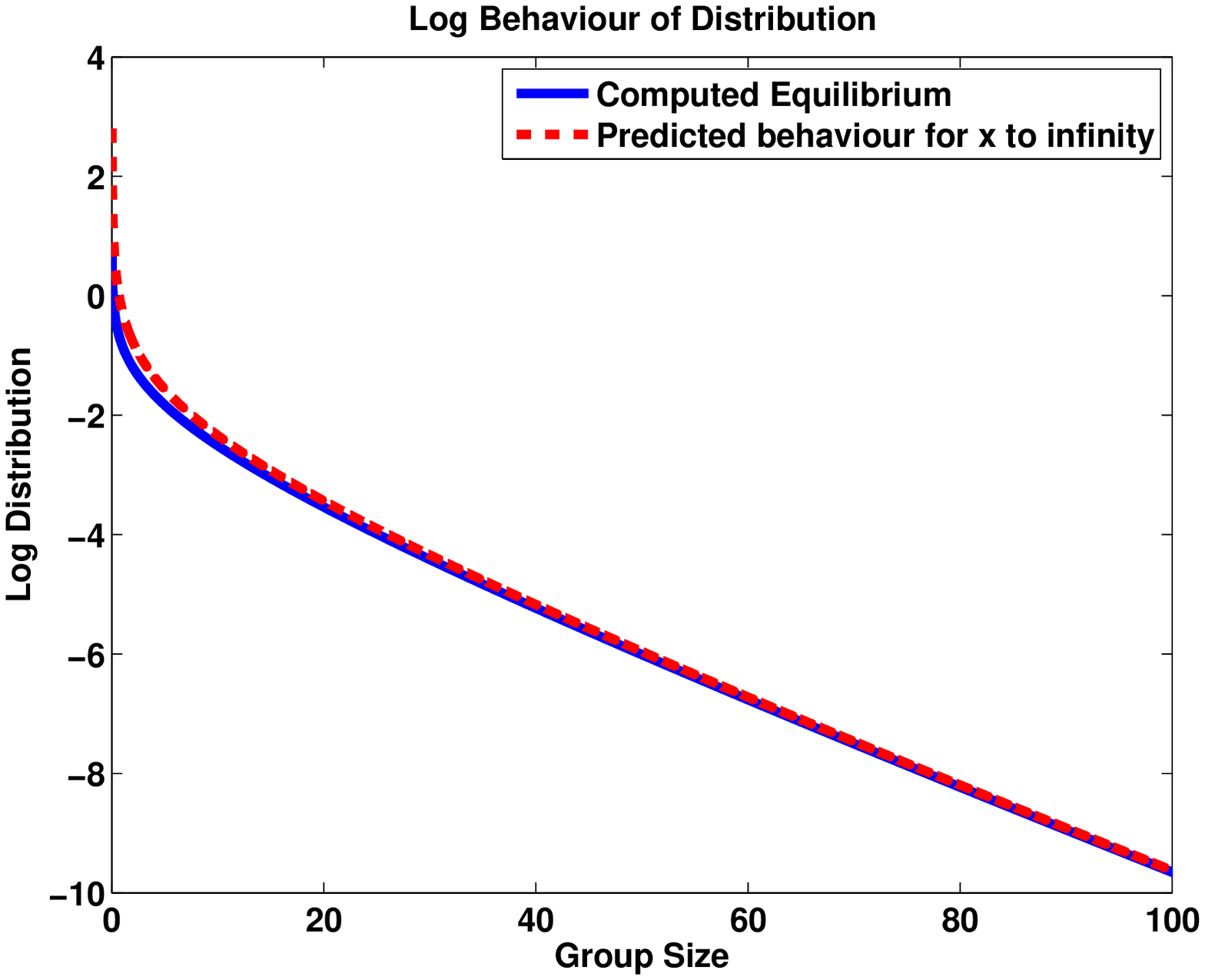}
  \caption{Log Distribution for $h=0.01$}
  \label{fig_5_6}
\end{subfigure}
\caption{The equilibrium distribution is approximated by the recursive scheme (Section~\ref{secrec}). In Fig.~\ref{fig_5_5}, the mass is $m_1^h =1$, the grid size $h=1$ and the equilibrium sequence is computed till $L=100$. It shows the generated distribution (blue solid line) in a log scale and presents the theoretically found asymptotic behaviour (red dashed line) in a log scale (Eq. \ref{reclargelog}) for the sake of comparison. One can observe perfect agreement for large sizes.
In Fig.~\ref{fig_5_6}, exactly the same is done for grid size $h=0.01$. Again, one can observe that the generated distribution shows the predicted asymptotics. }
\label{fig:rec}
\end{figure}
Observe that in both cases the equilibrium is non-negative. Note that the asymptotics are perfectly matched for both choices of $h$. As opposed to the truncated discretisation of the continuous model, one cannot observe any kink at the right-hand side of the graph. Obviously, this is the case since we don't need any truncation for the recursive algorithm. Additionally, one can observe that the large-size asymptotics differ for $h=1$ and $h=0.01$. We are going to investigate this phenomenon more precisely in the next section where we compare the large-size asymptotics of model D and model C.

\subsubsection{Link between discrete and continuous model}
\subsubsection*{i) Convergence for fixed interval length L} 
For $m_1 = m_1^h =1$ the continuous and discrete models can be compared as follows:
according to Eq. (\ref{flargelog}) the leading term in the asymptotics of the continuous equilibrium $f_{\infty}$ is given by $e^{-(4/27)x}$ as $ x \to \infty$. Set $x=hn$. Then, due to Eq. (\ref{reclargelog}), the leading term in the asymptotics of the discrete equilibrium $f_n^h$ is given by $[(1+\frac{4h}{27})^{-1/h}]^x$ as  $n (=x/h) \to \infty$. 
Since
\begin{equation} \label{largecomparison}
[(1+\frac{4h}{27})^{-1/h}]^x \to e^{-(4/27)x} \quad \text{as} \ h \to 0,
\end{equation}
the leading term of the discrete equilibrium converges to the leading term of the continuous equilibrium as $h \to 0$. Deploying the Newton method and the recursive scheme, we verify numerically if the same holds true for the truncated models uniformly on a fixed interval $[0,L]$.
Indeed, we can observe that for $h$ small enough and a fixed truncation size $L$, the discretized equilibrium for model C' (model D') as approximated by the Newton method and the equilibrium for model D generated by the recursive algorithm are very close. We have chosen $L=100$, $h=0.01$ and $m_1 = m_1^h =1$. The equilibrium computed by the Newton scheme -- the solid blue line in Fig.~\ref{fig_rec_newt} -- and the equilibrium computed by the recursive scheme -- the dotted red line in Fig.~\ref{fig_rec_newt} -- are the same up to a maximal absolute error of magnitude $10^{-6}$. This can be seen as an additional validation of the Newton method. \\\\
We have verified numerically uniform convergence of model C' and model D' in their large-size behaviour on finite intervals as $h \to 0$. This reflects the uniform convergence of model C and model D on finite intervals as indicated by Eq. (\ref{largecomparison}). We illustrate this by  fixing $L=200$ and comparing the asymptotics of model C and model D for $h$ becoming smaller.
\begin{figure}[htp]
\centering
\includegraphics[width=0.6\linewidth]{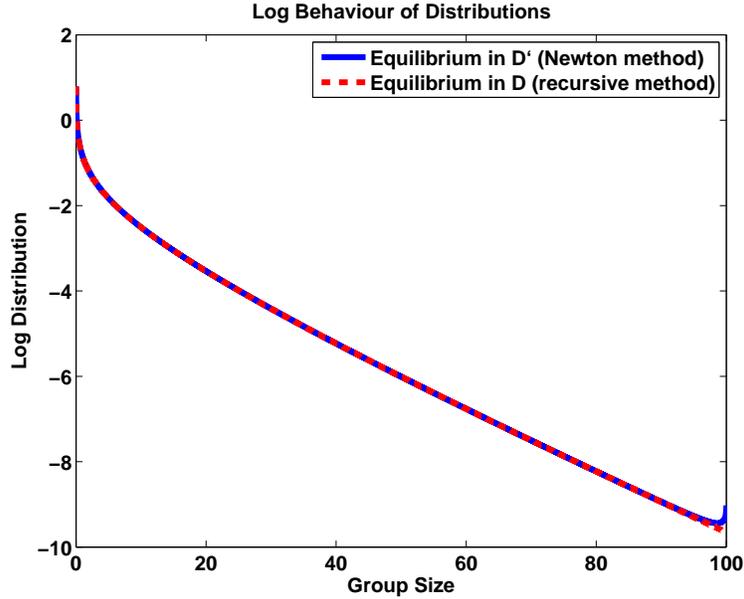}
\caption{Comparison of the equilibria for model D' and model D.  We take truncation $L=100$, grid size $h=0.01$ and mass $m_1 = m_1^h =1$. The equilibrium for model D' is generated by the Newton scheme (Section~\ref{secnewtmethod}) and represented in a log scale by the solid blue line. The equilibrium for model D is generated by the recursive scheme (Section~\ref{secrec}) and represented in a log scale by the dotted red line.}
\label{fig_rec_newt}
\end{figure}
Fig.~\ref{fig:uniform} shows the asymptotic large-size behaviour of the discrete equilibria $(f_i^h)_{ i \in \mathbb{N}}$ generated by the recursive algorithm in Section~\ref{secrec} and the analytically predicted continuous one (Eq. (\ref{flargelog})). We consider the grid sizes $h=1$, $h=0.1$ and $h=0.01$ and observe the expected convergence of both models.
\begin{figure}[htp]
\centering
  \includegraphics[width=0.6\linewidth]{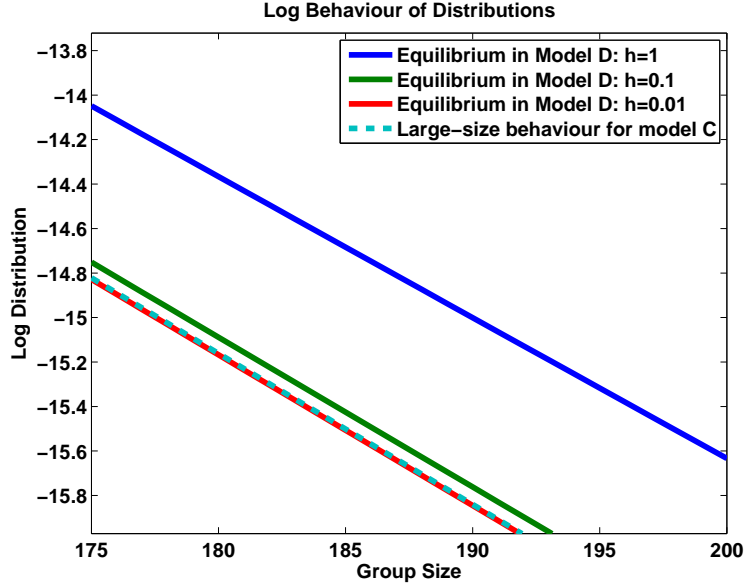}
\caption{The equilibrium distribution is generated by the recursive scheme, for mass $m_1 =1$, taking grid size $h=1$, $h=0.1$ and $h=0.01$. The figure shows the generated distributions (solid lines) and the large-size asymptotic behaviour for model C (dashed line) in a log scale (equation (\ref{flargelog})). We have magnified the plot close to $x=200$.}
\label{fig:uniform}
\end{figure}
The equilibrium in the genuine discrete case of model D, i.e. $h=1$, differs from the stationary solution of model C in its large-size behaviour. This difference becomes smaller for $h=0.1$ and even much smaller, invisible in the shown scale, for $h=0.01$.
\subsubsection*{ii) Divergence on increasing intervals}
If we fix $h$ and increase the investigated intervals of group sizes, the large-size behaviour of the discrete and continuous model diverge. We illustrate that in Fig.~\ref{fig:nonuniform} where we compare the predicted asymptotic behaviour for model D and model C at large group sizes $x$. The plots show the asymptotic behaviour close to $x=200$, $x=1000$, $x=2000$ for fixed $h=0.01$. One can see how the difference increases which means that for fixed $h$ the continuous and discrete equilibrium diverge as $ x \to \infty$:
\begin{figure}[htp]
\centering
\begin{subfigure}{.5\textwidth}
  \centering
  \includegraphics[width=1\linewidth]{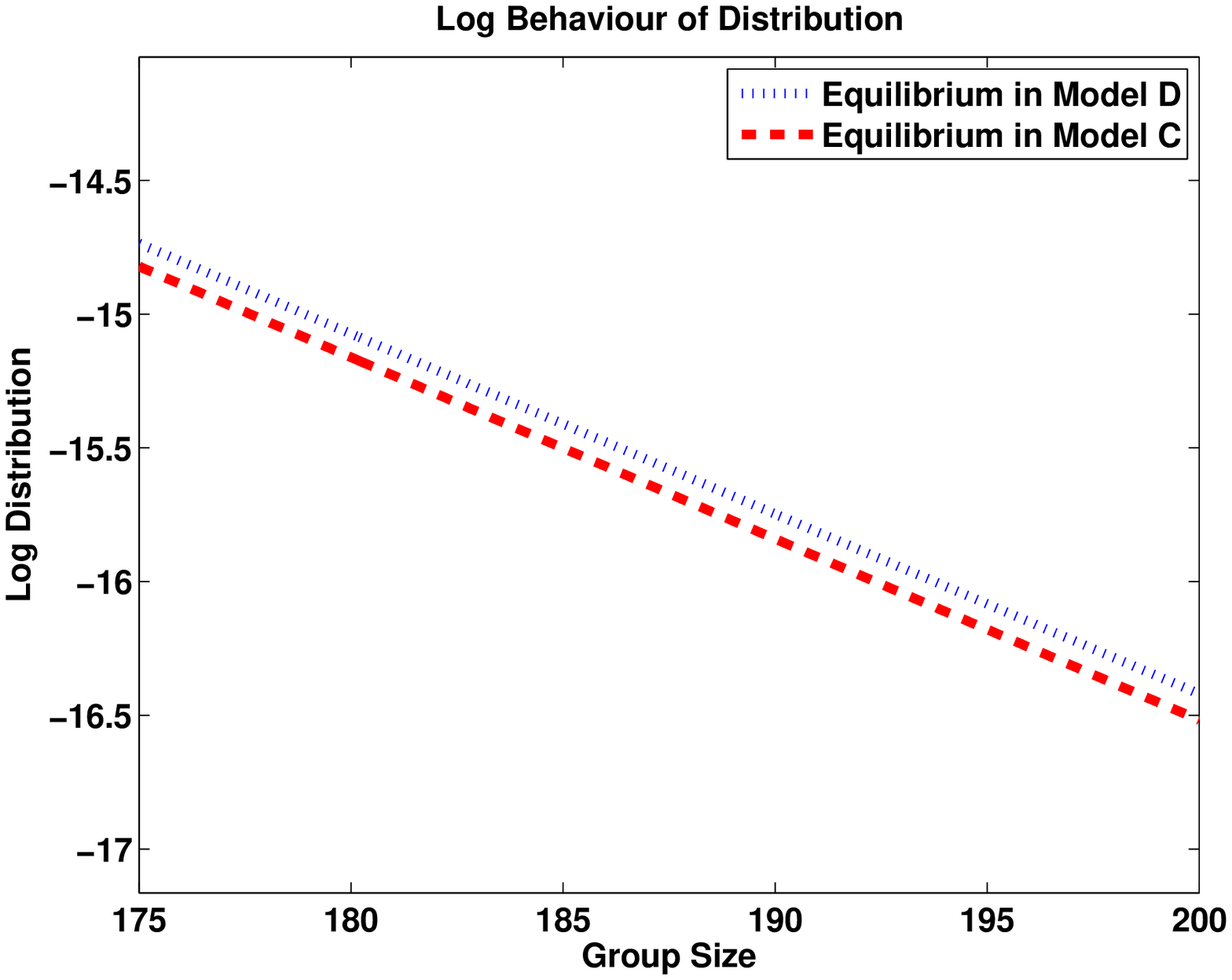}
  \caption{Log Distributions close to $x=200$}
  \label{fig_s_x}
\end{subfigure}%
\begin{subfigure}{.5\textwidth}
  \centering
  \includegraphics[width=1\linewidth]{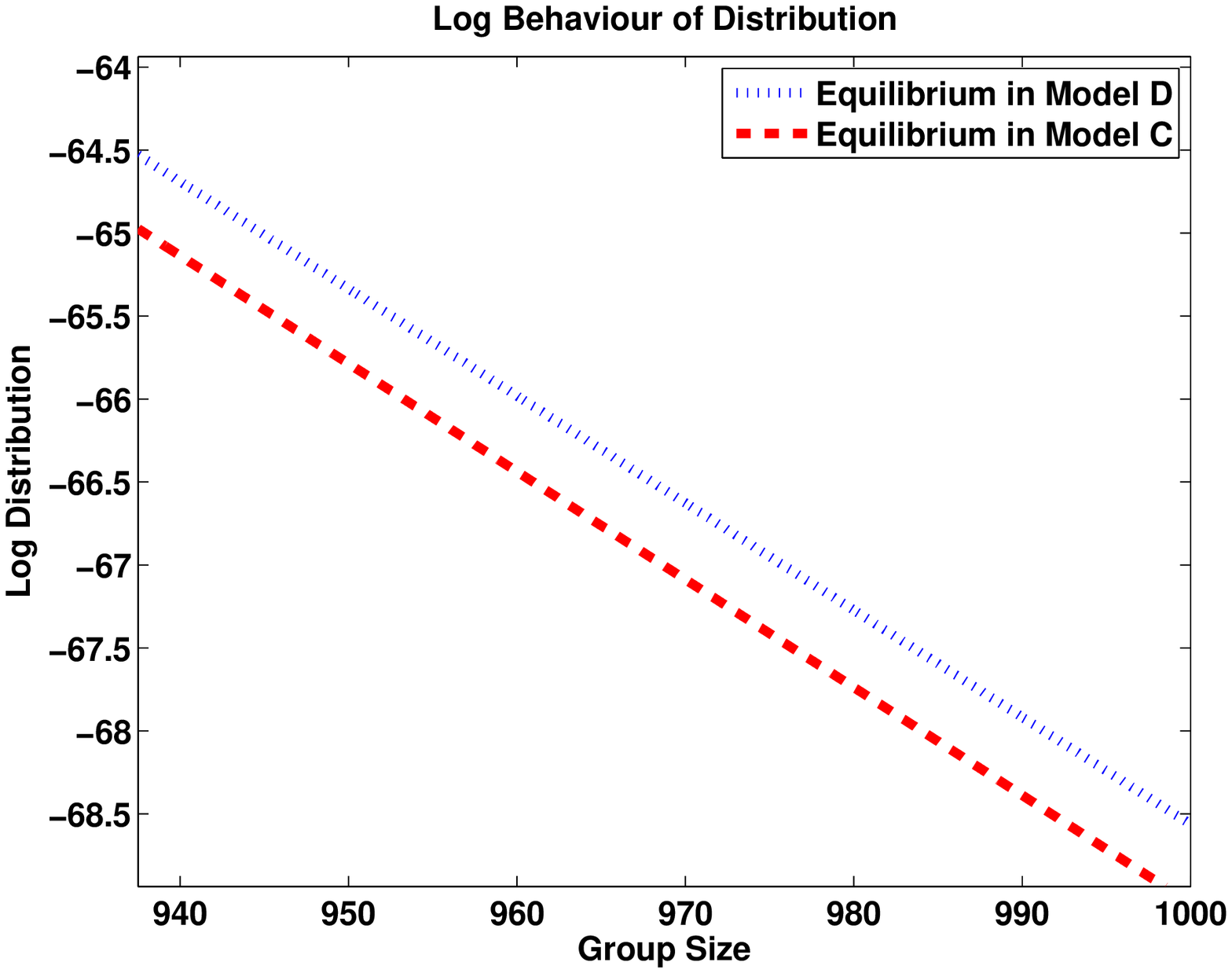}
  \caption{Log Distributions close to $x=1000$}
  \label{fig_m_x}
\end{subfigure}
\newline
\begin{subfigure}{.5\textwidth}
  \centering
  \includegraphics[width=1\linewidth]{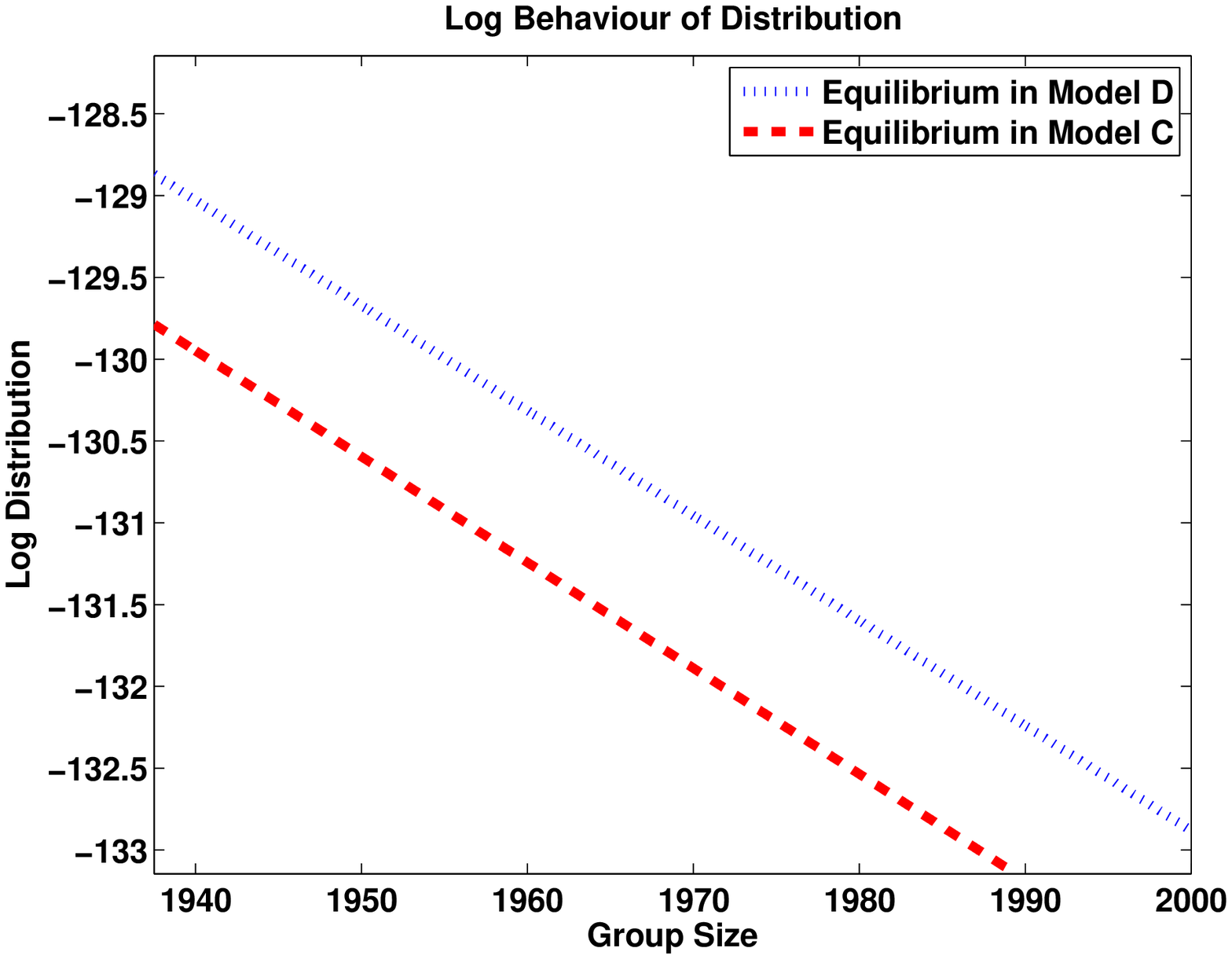}
  \caption{Log Distributions close to $x=2000$}
  \label{fig_l_x}
\end{subfigure}
\caption{The large-size behaviours of the discrete and continuous equilibrium distributions are compared, for mass $m_1 =1$ and fixed grid size $h=0.01$, close to $x=200$ (Fig.~\ref{fig_s_x}), close to $x=1000$(Fig.~\ref{fig_m_x}) and close to $x=2000$ (Fig.~\ref{fig_l_x}). In each case, it shows the large-size asymptotic behaviour for model D given by equation (\ref{reclargelog}) (blue dotted line) and the large-size asymptotic behaviour for model C given by equation (\ref{flargelog}) (red dashed line) in a log scale. Observe that for $x$ becoming greater, the difference between both graphs increases significantly.}
\label{fig:nonuniform}
\end{figure}
\subsection{Small-size behaviour for model D} \label{secsmallsize}
We turn towards the asymptotics of the equilibrium sequence in the case $h \to 0$. First, we need to investigate $m_0^h$ for $h \to 0$. As pointed out in \cite{DLP}[Section 15] we can immediately see from Eq.~(\ref{massnumber}) that the leading behaviour for $h \to 0$ is given by 
\begin{equation*}
m_0^h \sim 1 - (\frac{h}{m_1^h})^{1/3}.
\end{equation*}
Obviously, $f_1^h$ is a good indicator of the sought behaviour since it is the first term of the sequence. With the above and using (\ref{startofsequence})-(\ref{bsequence}), one gets (for taking $m_1^h =1$ in the end)
\begin{equation}
f_1^h = \frac{m_0^h(1-m_0^h)}{1 + 2m_0^h} \sim \frac{(1 - (\frac{h}{m_1^h})^{1/3})(\frac{h}{m_1^h})^{1/3}}{1 + 2(1- (\frac{h}{m_1^h})^{1/3})} \sim \frac{1}{3} \frac{h}{m_1^h}^{1/3} = \frac{1}{3} h^{1/3} \ \text{as} \ h \to 0.
\end{equation}
Let's compare this behaviour with the small-size asymptotics of the stationary solution of model C, denoted by $f$. We need to collate $f(h)$ with $ \frac{1}{h} f_1^h$ due to Eq.~(\ref{discretisation}). One can see that -- except for the factor $ \frac{1}{\Gamma(4/3)} \approx 1.12$ -- the discrete case actually has the same leading behaviour as the continuous one: 
\begin{align} \label{smallcomp}
\frac{1}{h} f_1^h &\sim \frac{1}{3} h^{-2/3} \ \text{as} \ h \to 0, \nonumber\\
f(h) &\sim \frac{1}{\Gamma(4/3)} \frac{1}{3} h^{- 2/3} \ \text{as} \ h \to 0  \ \text{(see (\ref{largesizegamma}))}.
\end{align}
In Fig.~\ref{fig:rec_smallFM} we compare $\frac{1}{h} f_1^h$ for $h \in [5*10^{-5}, 1]$ with the small-size behaviour of model C. We observe an approximation for decreasing $h$ due to the converging leading behaviour but the preservation of a small gap between the two graphs due to the different constants as seen in (\ref{smallcomp}).
\begin{figure}[htp]
  \centering
  \begin{subfigure}{.5\textwidth}
   \centering
  \includegraphics[width=1.0\linewidth]{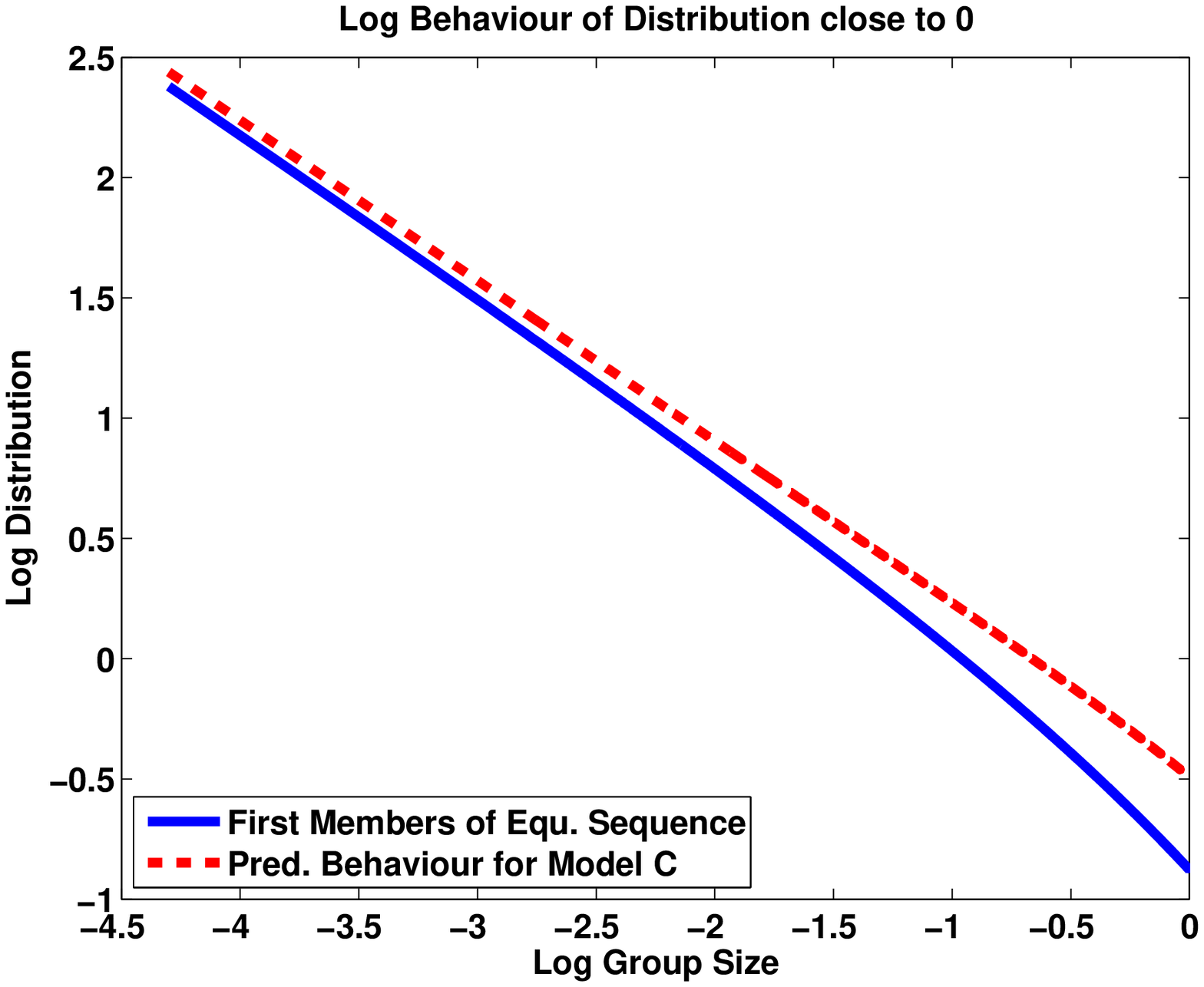}
  \caption{$\frac{1}{h}f_1^h$ for decreasing $h$}
   \label{fig:rec_smallFM}
  \end{subfigure}%
  \begin{subfigure}{0.5\textwidth}
  \centering
  \includegraphics[width=1.0\linewidth]{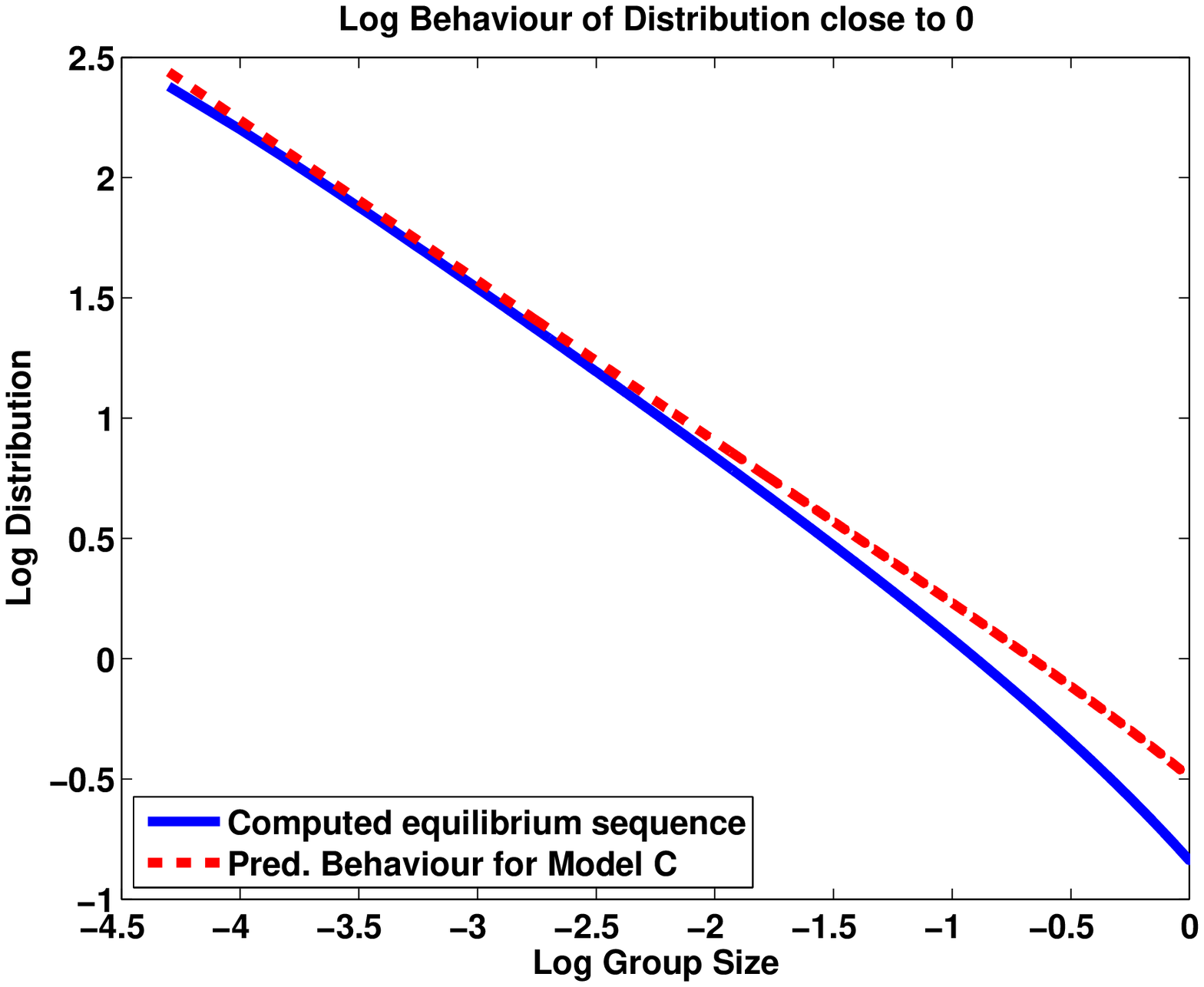}
  \caption{Equ. sequence for $h = 5*10^{-5}$}
  \label{fig:rec_small}
  \end{subfigure}
\caption{In Fig.~\ref{fig:rec_smallFM} we plot $\frac{1}{h} f_1^h$ for $h \in [5*10^{-5}, 1]$ in log-log scale (blue solid line) and the small-size asymptotics of the continuous model C (red dashed line). For small $h$, the graphs illustrate the findings in (\ref{smallcomp}).
In Fig.~\ref{fig:rec_small}, the equilibrium sequence for model D is generated as described in Section~\ref{secrec} taking mass $m_1^h =1$ and grid size $h = 5*10^{-5}$. The plot shows the distribution $(f_i^h)_{i \in \mathbb{N}} $ as a function of the group size in log-log scale (blue solid line) in the interval $[h,1]$ and the small-size asymptotics of the continuous model C (red dashed line). Both graphs tend to have the same slope for the sizes becoming smaller except for a slight divergence at the smallest group sizes.}
\end{figure}
In Fig.~\ref{fig:rec_small} we look at the equilibrium sequence given by the recursive algorithm for $h= 5*10^{-5}$, just in the interval $[h,1]$, and compare it to the behaviour predicted for the continuous case. We note that the two curves show a very close approximation for decreasing group sizes with the very first members of the sequence exhibiting the gap explained above. So the slope close to 0 becomes the same but diverges slightly for the first few members of the sequence. Again, this can be explained by model D providing a smaller fragmentation rate than model C, in connection with the fact that whereas the continuous equilibrium is defined on $(0,\infty)$ and has a singularity at $0$, the discrete equilibrium is a sequence.
\subsection{Determination of convergence rates}
Degond et al. have proven in \cite{DLP} that model C exhibits weak convergence to equilibrium as time goes to $\infty$. However, there is no finding about convergence almost everywhere. We want to show that the time-dependent solution $f(x,t)$ of Eq. (\ref{strongscaled}) converges uniformly to the equilibrium $ f_{\infty}$ if we start with a uniform distribution or also an exponential distribution. We also investigate the convergence rates for different group sizes. For simulating the convergence process, we work with the Euler method in the discretized version D' of the truncated model C'. Denote the discrete approximation of the time-dependent solution by $f_i(t)$ ( $\sim f(ih,t)$) and the discrete approximation of the equilibrium by $ f_i^{\infty}$. \\\\
Let's again choose the cut-off at $L=100$, grid size $h=0.01$, mass $m_1 = 1$ and time step size $dt=1$. As initial distribution we first take the uniform distribution (Table~\ref{tab:uniform}) as described in Section~\ref{sectimemethod} and then the exponential distribution (Table~\ref{tab:uniform_exp}) as for the Newton method, given by Eq. (\ref{exponential}).
The discretized equilibrium distribution $f_i^{\infty}$ is approximated by conducting the Euler scheme until $t=30$. Further, we calculate $f_i(25)$, $f_i(20)$, $f_i(15)$, $f_i(10)$ and $f_i(5)$ representing $f(x,25)$, \dots, $f(x,5)$. We evaluate the distributions at $i = 500,3500, 6500, 9500$ (representing $ x = 5, 35, 65, 95$)  and consider the relative distance to the equilibrium $ \frac{\left|f_i^{\infty} - f_i(t)\right|}{f_i^{\infty} } $ for $t = 5,10,15,20,25$ and $i = 500,3500, 6500, 9500$. Table~\ref{tab:uniform} gives an overview of the results for starting with a uniform distribution and Table~\ref{tab:uniform_exp} for starting with an exponential distribution.
\begin{table}[htp]
  \centering
    \begin{tabular}{ | l | l | l | l | l |}
    \hline
    Time $t$ & $x=5$ & $x=35$ &  $x=65$ &  $x=95$ \\ \hline
   $t=5$ & $0.2772$  & $9.2148$ & $154.7531$ & $2046.0000$ \\ \hline
    $t=10$ & $0.0638$ & $1.4009$ & $10.8288$ & $67.0145$ \\ \hline
    $t=15$ & $0.0089$ & $0.1976$ & $1.0721$ & $3.8651 $ \\ \hline
    $t=20$ & $0.0012$ & $0.0260$ & $0.1300$ & $0.3832$ \\ \hline
    $t=25$ & $0.0001$ & $0.0030$ & $0.0149$ & $0.0423$ \\ \hline
    \end{tabular}
    \vskip 0.3 cm
\caption{Starting with a uniform distribution the time-dependent solution of model C, $f(x,t)$, is approximated via the Euler scheme for model D', taking $L=100$, $h=0.01$, $dt=1$ and $m_1 =1$. This approximation, $f_i(t)$, is evaluated at $t = 5,10,15,20$ and the equilibrium distribution is approximated via following the Euler scheme until $t=30$. The table shows the relative distances to the equilibrium, $ \frac{\left|f_i^{\infty} - f_i(t)\right|}{f_i^{\infty} } $ for $t = 5,10,15,20,25$ and $i = 500, 3500, 6500, 9500$.}
  \label{tab:uniform}
\end{table}
\begin{table}[htp]
  \centering
    \begin{tabular}{ | l | l | l | l | l |}
    \hline
    Time $t$ &  $x=5$ &  $x=35$ &  $x=65$ &  $x=95$ \\ \hline
   $t=5$ & $0.05620$  & $0.75370$ & $0.98890$ & $0.99980$ \\ \hline
    $t=10$ & $0.00800$ & $0.16010$ & $0.50000$ & $0.77500$ \\ \hline
    $t=15$ & $0.00120$ & $0.02670$ & $0.11420$ & $0.25710 $ \\ \hline
    $t=20$ & $0.00020$ & $0.00430$ & $0.02000$ & $0.05220$ \\ \hline
    $t=25$ & $0.00003$ & $0.00004$ & $0.00290$ & $0.00790$ \\ \hline
    \end{tabular}
    \vskip 0.3 cm
\caption{Starting with an exponential distribution  the time-dependent solution of model C, $f(x,t)$, is approximated via the Euler scheme for model D', taking $L=100$, $h=0.01$, $dt=0.5$ (smaller than in the previous case due to stabilisation problems for small sizes) and $m_1 =1$. This approximation, $f_i(t)$, is evaluated at $t = 5,10,15,20$ and the equilibrium distribution is approximated via following the Euler scheme until $t=30$. The table shows the relative distances to the equilibrium, $ \frac{\left|f_i^{\infty} - f_i(t)\right|}{f_i^{\infty} } $ for $t = 5,10,15,20,25$ and $i = 500, 3500, 6500, 9500$.}
  \label{tab:uniform_exp}
\end{table}
The tables indicate that the convergence is uniform on a bounded interval since the distance to equilibrium decreases in time monotonically for any $i$ (resp. $x$). Taking a uniform initial distribution effects in the relative distances being on a much smaller scale for small $i$ than for large $i$. The impact of the initial distribution vanishes on the long run and the convergence rates seem to become the same for different group sizes.\\\\
We are investigating the speed of convergence depending on the sizes more thoroughly.
Consider the following approach for determining the exponential convergence rate $\delta_{x,t}$ where $x$ stands for the group size and $t$ for time: one can express $f(x,t)$ as
\begin{equation*}
f(x,t) = f_{\infty}(x)(1- e^{-t \delta_{x,t} }) + f_0(x) e^{-t \delta_{x,t} }.
\end{equation*}
Substracting and dividing both sides by $f_{\infty}(x)$ and taking absolute values gives
\begin{equation}
\mu(x,t) := \frac{\left|f(x,t)-f_{\infty}(x)\right|}{f_{\infty}(x)}  = \frac{\left|f_0(t)-f_{\infty}(x)\right| f_0(x) e^{- t \delta_{x,t} }}{f_{\infty}(x)}. 
\end{equation}
Hence, for two different points of time $t_1$ and $t_2$, one gets 
\begin{equation*}
\frac{\mu(x,t_2)}{\mu(x,t_1)} = e^{-(t_2 \delta_{x,t_2}  - t_1 \delta_{x,t_1})}.
\end{equation*}
Thus, if the convergence rate is the same for $t_2$ and $t_1$, it can be expressed as
\begin{equation} \label{convergencerate}
\delta_{x,t_1} = \delta_{x,t_2} = \frac{1}{t_2 - t_1} \log { (\frac{\mu(x,t_1)}{\mu(x,t_2})}.
\end{equation}
We have estimated $\delta_{x, t_2}$ numerically for $x = 5$ and $x = 95$ by calculating the relative distances $\mu(x,t_1), \mu(x,t_2)$ as for Table~\ref{tab:uniform} and Table~\ref{tab:uniform_exp}. The points of time $t_1$, $t_2$ were taken to be $t_1 =20, \dots, 28$ and $t_2 = t_1 +1$.  We have started with a uniform distribution (Fig.~\ref{fig_rate_uni}) and with an exponential distribution (Fig.~\ref{fig_rate_exp}) and observed -- as expected -- the same limit behaviour for the convergence rates.
\begin{figure}[htp]
\centering
\begin{subfigure}{.5\textwidth}
  \centering
  \includegraphics[width=1\linewidth]{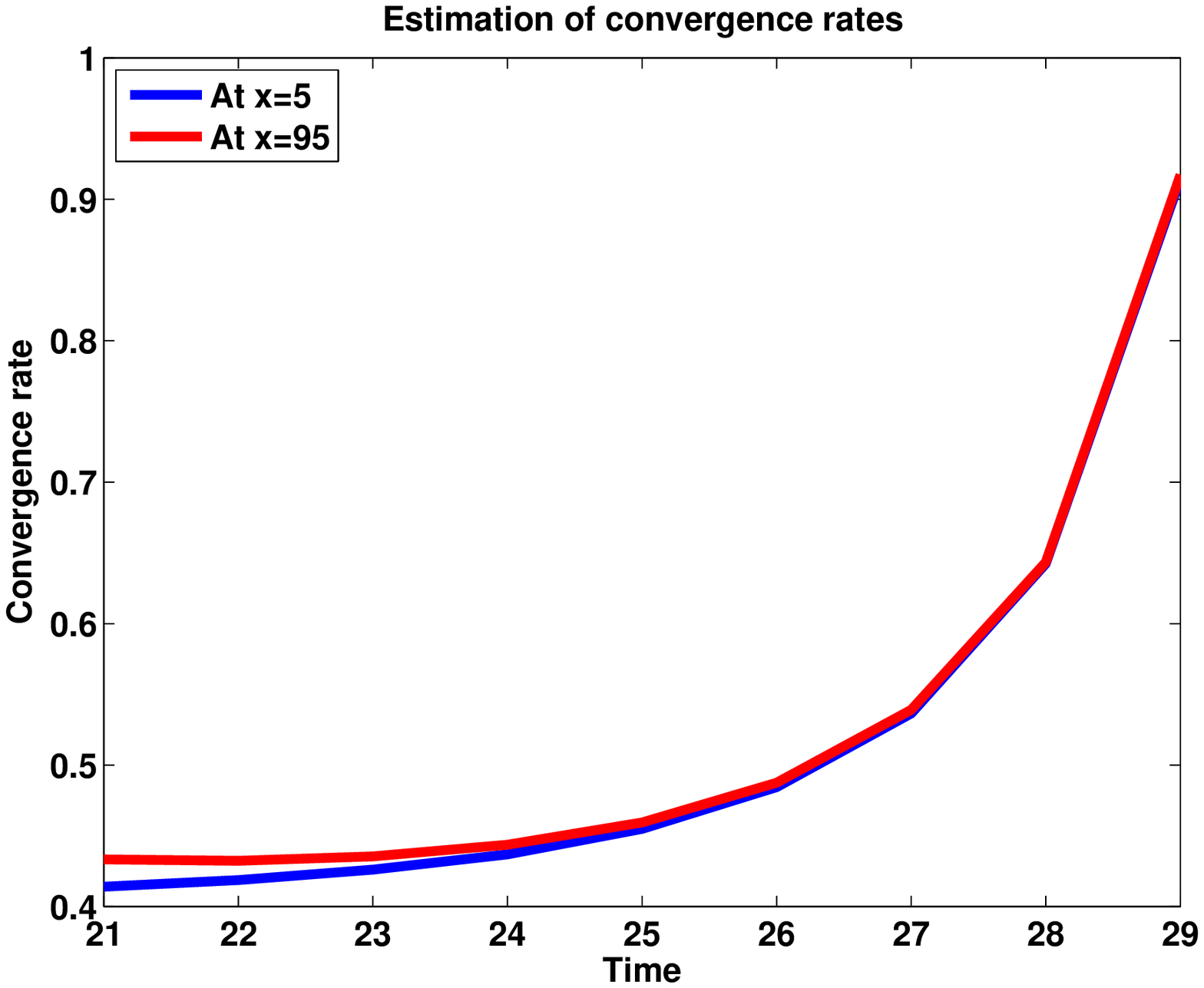}
  \caption{Convergence rates for uniform initial}
  \label{fig_rate_uni}
\end{subfigure}%
\begin{subfigure}{.5\textwidth}
  \centering
  \includegraphics[width=1\linewidth]{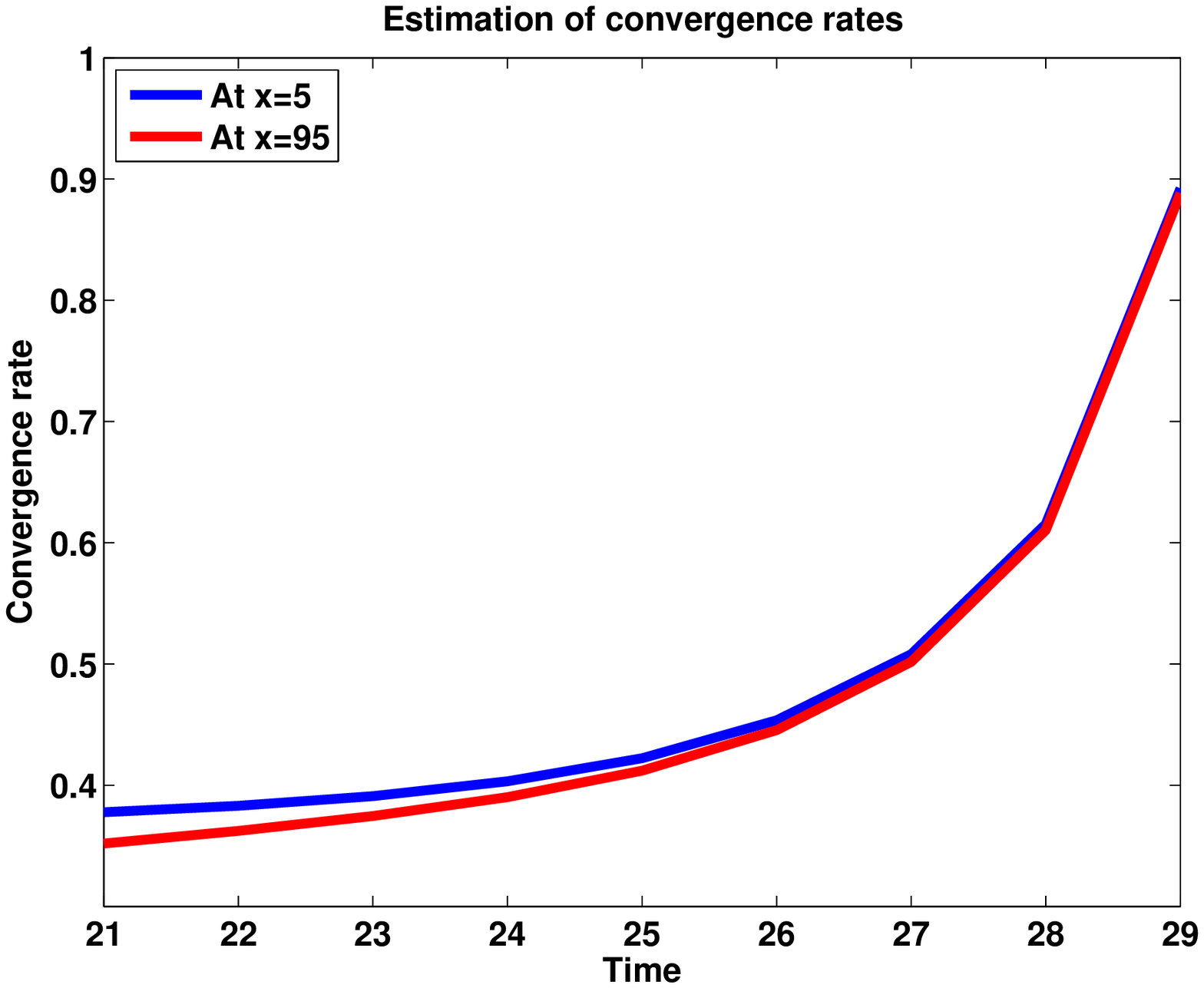}
  \caption{Convergence rates for exponential initial}
  \label{fig_rate_exp}
\end{subfigure}
\caption{Starting with a uniform distribution (Fig.~\ref{fig_rate_uni}) and with an exponential distribution (Fig.~\ref{fig_rate_exp}), the time-dependent solution of model C, $f(x,t)$, is approximated via the Euler scheme for model D', taking $L=100$, $h=0.01$, $m_1 =1$ and $dt=1$ for uniform initial and $dt = 0.5$ for exponential initial (due to stability issues for small sizes). The approximation, $f_i(t)$, is evaluated at $t = 20,\dots, 29$ and the equilibrium distribution is approximated via following the Euler scheme until $t=30$. Calculating the relative distances to the equilibrium, $ \mu_i(t) = \left|f_i^{\infty} - f_i(t)\right|/f_i^{\infty} $, for $i = 500$ and $i=9500$ (representing $x=5$ and $x=95$), we estimate the exponential convergence rate $\delta_{x,t_2}$ ($\sim \delta_{x,t_1}$) for $t_1 = 20,\dots, 28$ and $t_2 = t_1 +1$ according to Eq. (\ref{convergencerate}).}
\label{fig:rate}
\end{figure}
Note that in both cases the estimated convergence rates become the same for the small and the large size. The increase in time indicates super-exponential convergence rates.
\section{Conclusion} \label{Conclusion}
In this work, we have investigated numerically the coagulation-fragmentation model for animal group size distributions theoretically discussed by Degond et al. in \cite{DLP}. The central point of this work was to approximate the equilibria numerically and investigate convergence to equilibrium. We have worked with three different numerical methods: a recursive algorithm -- first introduced by Ma et al. in \cite{MJS} --  and a Newton and a time-dependent method -- developed in this paper. 
We have validated our numerical methods by checking the accordance with the predicted asymptotic behaviour and used the time-dependent scheme to show that there is super-exponential convergence to equilibrium in time on finite intervals.\\\\
We have seen that the Newton method provides a very fast approximation of the equilibrium after just five iterations. We suggest that the algorithm could be used in more complicated models with coagulation and fragmentation rates depending on the group sizes and/or time. Further, the Newton scheme could be deployed to prove the existence and uniqueness of the equilibrium in such models where the Bernstein method -- used in \cite{DLP} -- fails as it solely works for fixed coagulation and fragmentation parameters. Another topic of possible future work is to analyse the indicated super-exponential convergence more precisely and determine the convergence rates analytically.
\section*{Acknowledgments} The authors would like to thank J-G. Liu and R. Pego for enlightening discussions. 
This work has been supported by the Engineering and Physical Sciences Research Council (EPSRC) under grant ref. EP/M006883/1, and by the National Science Foundation (NSF) under grant
RNMS11-07444 (KI-Net). P. D. is on leave from CNRS, Institut de Math\'ematiques, Toulouse, France. He acknowledges support from the Royal Society and the Wolfson foundation through a Royal Society Wolfson Research Merit Award. M.E. has been supported by the German National Academic Foundation during the first part of this work and is now supported by a Roth Scholarship from the Department of Mathematics at Imperial College London.

% You may incorporate your references as follows in your main tex file.
% Using BibTex is not recommended but can be handled.

\end{document}